\documentclass[12pt]{article}%9/22
\usepackage{amsmath}
\usepackage{amsfonts}
\usepackage{amssymb}
\usepackage{graphicx}%
\providecommand{\U}[1]{\protect\rule{.1in}{.1in}}
\newcommand{\Lie}{\mathop{\rm Lie}}
\newcommand{\Hom}{\mathop{\rm Hom}}
\newcommand{\Ad}{\mathop{\rm Ad}}
\newtheorem{theorem}{Theorem}

\newtheorem{lemma}[theorem]{Lemma}

\newtheorem{proposition}[theorem]{Proposition}
\newtheorem{remark}[theorem]{Remark}

\newenvironment{proof}[1][Proof]{\noindent\textbf{#1.} }{\ \rule{0.5em}{0.5em}}
\begin{document}
\date{}

\title{Action of the conformal group on steady state solutions to Maxwell's
equations and background radiation}
\author{Bertram Kostant and Nolan Wallach}
\maketitle

\begin{abstract}
The representation of the conformal group ($PSU(2,2)$) on the space of solutions to Maxwell's equations on the conformal compactification of 
Minkowski space is shown to break up into four irreducible unitarizable smooth Fr\'echet representations of moderate growth. An explicit
inner product is defined on each representation. The frequency spectrum of each of these representations is analyzed. These representations have notable properties; in particular they have positive or negative energy, they are of type  $A_{\frak q}(\lambda)$ and are quaternionic. Physical implications of the results are explained.
\end{abstract}

\section{Introduction}

\let\thefootnote\relax\footnotetext{Key words and phrases: Maxwell's equations, conformal
compactification, conformal group, unitary representation}
\let\thefootnote\relax\footnotetext{MSC 2010 : 78A25, 58Z05, 22E70, 22E45.}
\let\thefootnote\relax\footnotetext{Research partially supported by NSF grant DMS 0963035.}
The purpose of this
paper is to analyze the steady state solutions of Maxwell's equations in a
vacuum using the tools of representation theory. By steady state we mean those
solutions that extend to the conformal compactification of Minkowski space.
That is, we look upon the solutions of Maxwell's equations as tensor valued on

$\mathbb{R}^{4}$ with the flat Lorentzian metric given by
\[
-dx_{1}^{2}-dx_{2}^{2}-dx_{3}^{2}+dt^{2},
\]
where $x_{1},x_{2},x_{3}$ and $x_{4}=t$ yield the standard coordinates of $%
\mathbb{R}^{4}$. $
\mathbb{R}
^{4}$ with this metric will be denoted by 
$\mathbb{R}^{1,3}$. The conformal compactification is the space $S^{3}\times_{\{\pm1
\}}S^{1}$ (modulo the product action of $\pm1$) with (up to a positive scalar
multiple on both factors) the product metric with the negative of the constant
curvature 1 metric on $S^{3}$ and the usual metric on $S^{1}$. The injection,
$f$, of $
\mathbb{R}^{4}$ into $S^{3}\times_{\{\pm1\}}S^{1}$ is the inverse of a variant of
stereographic projection (see Section 2). This embedding is not an isometry,
but it is conformal. Our approach to Maxwell's equations uses the equivalent
formulation in terms of differential $2$-forms on Minkowski space. More
generally, if $(M,g)$ is an oriented Lorentzian four manifold (signature
$(-,-,-,+)$) and if \textquotedblleft\ $\ast$\ \textquotedblright\ denotes the
Hodge star operator on $2$-forms relative to the volume form, $\gamma$, with
$g(\gamma,\gamma)=-1$, then there is a version of Maxwell's equations on
$\Omega^{2}(M)$ (differential two forms) given by%
\[
d\omega=d\ast\omega=0.
\]
Let $M$ and $N$ be four-dimensional Lorentzian manifolds, let
$F:M\rightarrow N$ be a conformal transformation, and let $\omega$ be a
solution to Maxwell's equations on $N$. Then $F^{\ast}\omega$ is a solution to Maxwell's equations on $M$. Since $f$ is conformal we see that the pullback of
solutions to Maxwell's equations on $S^{3}\times_{\{\pm1\}}S^{1}$ yields
solutions to the usual Maxwell equations.

The group of conformal transformations of $S^{3}\times_{\{\pm1\}}S^{1}$ is
locally the group $SO(4,2)$ and thus the solutions to Maxwell's equations on
$S^{3}\times_{\{\pm1\}}S^{1}$ form a representation of this group. We can
interpret this as follows: We first note that we can replace $S^{3}%
\times_{\{\pm1\}}S^{1}$ with the group $U(2)$. If on $\Lie(U(2))$ we put the
Lorentzian form that corresponds to the quadratic form $-\det X$, then the
corresponding bi-invariant metric on a $U(2)$ is isometric, up to positive
scalar multiple, with $S^{3}\times_{\{\pm1\}}S^{1}$. We interpret this space
as the Shilov boundary of the Hermitian symmetric space that corresponds to
$G=SU(2,2)$ (which is locally isomorphic with $SO(4,2)$).

We denote by Maxw the space of solutions in $\Omega^{2}(U(2))$ to Maxwell's
equations. We show that there is a canonical nondegenerate $G$-invariant
Hermitian form on Maxw. Further, we show that as a smooth Fr\'{e}chet
representation of $G$, Maxw splits into the direct sum of four irreducible
(Fr\'{e}chet) representations (of moderate growth) that are mutually
orthogonal relative to the form. This form is positive definite on two of the
irreducible pieces and negative definite on the other two. Since the
$K=S(U(2)\times U(2))$ isotypic components of Maxw are all finite dimensional
we see that this yields four unitary irreducible representations of $G$. Two
of the representations are holomorphic (negative energy in the physics
literature) and two are anti-holomorphic (positive energy). We also describe
them in terms of the $A_{\mathfrak{q}}(\lambda)$ that yield second continuous
cohomology (the four theta stable parabolics $\mathfrak{q}$ involved relate
these representations to twister theory) and in terms of quaternionic
representations ($SU(2,2)$ is the quaternionic real form of $SL(4,\mathbb{C})$). These representations are actually representations of $PSU(2,2)$. In this group there is a dual pair
$PSU(1,1),SO(3)$ establishing an analogue of Howe duality in each of the four representations. The
realization of these representations is intimately related to the work in [K].We use the decomposition
of the restriction of these representations to $PSU(1,1)$ to analyze the frequency distribution of the solutions
in each of the $PSU(2,2)$ representations.

We interpret the plane wave solutions as
generalized Whittaker vectors on Maxw and the solutions as wave packets of the
Whittaker vectors. These wave packets have constrained frequency spectrum and
using Planck's black body radiation law the frequency limitation and
luminosity of the corresponding radiation determines the temperature to narrow
constraints. This means that we can fit our solutions to the measured
background radiation on a steady state universe. We make no assertions as to
how such a steady state universe might physically exist. There are many
suggestions in the literature (e.g., the work of Hoyle et al [HBN]). All seem
complicated. However, we will content ourselves to the assertion that the big
bang is not necessarily the only possible interpretation of background radiation.

There is also an interpretation of the red shift that can be gleaned from this
work involving the relationship between the measurement of time from the
proposed \textquotedblleft big bang\textquotedblright\ and the steady state
\textquotedblleft time\textquotedblright\ which is periodic but with a large
period appearing to move faster as we look backwards or forwards in terms of
\textquotedblleft standard\textquotedblright\ time. (See [S].)

We are aware that many of the aspects of representation theory in this paper
could have been done in more generality. We have constrained our attention to
the four representations at hand since the main thrust of this paper is to show
how representation theory can be used to study well-known equations in physics.

Parts of this work should be considered  expository. Related work has
been done by [HSS] on the action of the conformal group on solutions of the
wave equation and [EW] relating positive energy representations to generalized Dirac equations.

\section{Conformal compactification of Minkowski space}

Let $
\mathbb{R}^{1,3}$ denote $\mathbb{R}^{4}$ with the pseudo-Riemannian (Lorentzian) structure given by
$(x,y)=-x_{1}y_{1}-x_{2}y_{2}-x_{3}y_{3}+x_{4}y_{4}$. Here $x_{i},i=1,2,3,4,$
are the standard coordinates on $\mathbb{R}^{4}$ and we identify the tangent space at every point with
$\mathbb{R}^{4}$. We can realize this space as the space of $2\times2$ Hermitian matrices
(a $4$-dimensional vector space over
$\mathbb{R}$), $V$, with the Lorentzian structure corresponding to the quadratic form
given by the determinant. Note that%
\[
\det\left[
\begin{array}
[c]{cc}%
x_{4}+x_{3} & x_{1}+ix_{2}\\
x_{1}-ix_{2} & x_{4}-x_{3}%
\end{array}
\right]  =(x,x).
\]
We can also realize the space in terms of skew Hermitian matrices
$\mathfrak{u}(2)=iV$ and noting that the form becomes $-\det$. With this
interpretation, and realizing that $\mathfrak{u}(2)=\Lie(U(2))$ we have an
induced Lorentzian structure, $\left\langle \ldots,\ldots\right\rangle $ on $U(2)$.
We also have a transitive action of $K=U(2)\times U(2)$ on $U(2)$ by right and
left translation%
\[
(g_{1},g_{2})u=g_{1}ug_{2}^{-1}.
\]
Since the isotropy group at $I$ is $M=\rm diag(U(2))=\{(g,g)|g\in U(2)\}.$We see
that $K$ acts by isometries on $U(2)$ with this structure.

We will now consider a much bigger group that acts. We first consider the
indefinite unitary group $G\cong U(2,2)$ given by the elements, $g\in M_{4}(
\mathbb{C})$ such that
\[
g\left[
\begin{array}
[c]{cc}%
0 & iI_{2}\\
-iI_{2} & 0
\end{array}
\right]  g^{\ast}=\left[
\begin{array}
[c]{cc}%
0 & iI_{2}\\
-iI_{2} & 0
\end{array}
\right]  .
\]
Here as usual $g^{\ast}$ means conjugate transpose. $G\cap U(4)$ is the group
of all matrices of the $2\times2$ block form%
\[
\left[
\begin{array}
[c]{cc}%
A & B\\
-B & A
\end{array}
\right]
\]
satisfying $AB^{\ast}=BA^{\ast}$ and $AA^{\ast}+BB^{\ast}=I$. These equations
are equivalent to the condition
\[
A+iB\in U(2).
\]
It is easy to see that the map
\[
\Psi:G\cap U(4)\rightarrow U(2)\times U(2)
\]
given by%
\[
\left[
\begin{array}
[c]{cc}%
A & B\\
-B & A
\end{array}
\right]  \mapsto\left(  A-iB,A+iB\right)
\]
defines a Lie group isomorphism. This leads to the action of $G\cap U(4)$ on
$U(2)$ given by%
\[
\left[
\begin{array}
[c]{cc}%
A & B\\
-B & A
\end{array}
\right]  \cdot x=(A-iB)x(A+iB)^{-1}.
\]
Note that the stabilizer of $I_{2}$ is the subgroup isomorphic with $M$ given
by the elements
\[
\left[
\begin{array}
[c]{cc}%
A & 0\\
0 & A
\end{array}
\right]  ,A\in U(2).
\]
We extend this to a map of $G$ to $U(2)$ given by%
\[
\Phi:\left[
\begin{array}
[c]{cc}%
A & B\\
C & D
\end{array}
\right]  \mapsto(A+iC)(A-iC)^{-1}.
\]
This makes sense since $A-iC$ is invertible if $\left[
\begin{array}
[c]{cc}%
A & B\\
C & D
\end{array}
\right]  \in G$. We consider the subgroup $P$ of $G$ that consists of the
matrices%
\[
\left[
\begin{array}
[c]{cc}%
g & gX\\
0 & (g^{\ast})^{-1}%
\end{array}
\right]
\]
with $g\in GL(2,
\mathbb{C})$ and $X\in H$ (in other words, $X^{\ast}=X).$ Then every element of $G$ can
be written in the form $kp$ with $k\in G\cap U(4)$ and $p\in P$. We note that
$\Phi(kp)=\Phi(k)=k\cdot I$. Now, $U(4)\cap K$ acts transitively on $G/P$ and
the stabilizer of the identity coset is the group $\left[
\begin{array}
[c]{cc}%
A & 0\\
0 & A
\end{array}
\right]  ,A\in U(2)$. We will identify $K$ with $U(4)\cap G$ (under $\Psi$)
and $M$ with the stabilizer of the identity. Thus $G/P=K/M$.

We consider the subgroup $\overline{N}$:%
\[
\left[
\begin{array}
[c]{cc}%
I & 0\\
Y & I
\end{array}
\right]  ,Y^{\ast}=Y.
\]
If we write%
\[
\left[
\begin{array}
[c]{cc}%
I & 0\\
Y & I
\end{array}
\right]  =kp
\]
with
\[
k=\left[
\begin{array}
[c]{cc}%
A & B\\
-B & A
\end{array}
\right]
\]
as above, then $A\in GL(2,\mathbb{C})$ and
\[
-BA^{-1}=Y.
\]
One can see that if we set%
\[
k(Y)=\left[
\begin{array}
[c]{cc}%
\frac{I}{\sqrt{I+Y^{2}}} & \frac{-Y}{\sqrt{I+Y^{2}}}\\
\frac{Y}{\sqrt{I+Y^{2}}} & \frac{I}{\sqrt{I+Y^{2}}}%
\end{array}
\right],
\]
then
\[
k(Y)P=\left[
\begin{array}
[c]{cc}%
I & 0\\
Y & I
\end{array}
\right]  P.
\]
This gives an embedding of $H$ into $U(2)$%
\[
Y\longmapsto(I+iY)(I-iY)^{-1},
\]
the Cayley transform. We next explain how this is related to the Cayley
transform in the sense of bounded symmetric domains.

We note that it is more usual to look upon $G$ (in its more usual incarnation)
as the group of all elements $g\in GL(4,\mathbb{C})$ such that%
\[
g\left[
\begin{array}
[c]{cc}%
I & 0\\
0 & -I
\end{array}
\right]  g^{\ast}=\left[
\begin{array}
[c]{cc}%
I & 0\\
0 & -I
\end{array}
\right]  .
\]
Let us set $G_{1}$ equal to this group. The relationship between the two
groups is given as follows. Set%
\[
L=\frac{1}{\sqrt{2}}\left[
\begin{array}
[c]{cc}%
I & iI\\
I & -iI
\end{array}
\right]
\]
(a unitary matrix) if
\[
\sigma(g)=LgL^{\ast}%
\]
then $\sigma$ defines an isomorphism of $G$ onto $G_{1}$. $G_{1}$ has an
action by linear fractional transformations on the bounded domain,
$\mathbf{D}$, given as the set of all $Z\in M_{2}(\mathbb{C})$
 such that $ZZ^{\ast}<I$ (here $<$ is the order defined by the cone of
positive definite Hermitian matrices. If $g=\left[
\begin{array}
[c]{cc}%
A & B\\
C & D
\end{array}
\right]  \in G_{1}$, then%
\[
g\cdot Z=(AZ+B)(CZ+D)^{-1}.
\]
We note that if $Y^{\ast}=Y$ then
\[
\sigma\left(  \left[
\begin{array}
[c]{cc}%
I & 0\\
Y & I
\end{array}
\right]  \right)  \cdot I=(I+iY)(I-iY)^{-1}%
\]
and%
\[
\sigma\left(  \left[
\begin{array}
[c]{cc}%
I & Y\\
0 & I
\end{array}
\right]  \right)  \cdot I=I.
\]
The embedding $F$ of Minkowski space 
$\mathbb{R}^{1,3}$ into $U(2)$ given by%
\[
(x_{1},x_{2},x_{3},x_{4})\mapsto\left[
\begin{array}
[c]{cc}%
x_{4}+x_{3} & x_{1}+ix_{2}\\
x_{1}-ix_{2} & x_{4}-x_{3}%
\end{array}
\right]  =X\longmapsto(I+iX)(I-iX)^{-1}%
\]
embeds it as a dense open subset. However, it is only a conformal embedding. Indeed

\begin{lemma}
The embedding $F$ is conformal with
\[
(F^{\ast}\left\langle \ldots,\ldots\right\rangle )_{x}=4\left(  1+2\sum x_{i}%
^{2}+(x,x)^{2}\right)  ^{-1}(\ldots,\ldots)_{x}.
\]

\end{lemma}

\begin{proof}
We note that $T_{u}(U(2))=\{uX|X\in\mathfrak{u}(2)\}$. Furthermore,
$\left\langle uX,uX\right\rangle _{u}=-\det(X)$. Now let $Y\in M_{2}(\mathbb{C})$ 
be such that $Y^{\ast}=Y$, that is $Y\in H$. Let $Q(Y)=$ $(I+iY)(I-iY)^{-1}%
$. We calculate $\left\langle dQ_{Y}(v),dQ_{Y}(v)\right\rangle _{Q(Y)}$ for
$v\in H$ thought of as being an element of $T_{Y}(H)$. We get%
\[
dQ_{Y}(v)=iv(I-iY)^{-1}+(I+iY)(I-iY)^{-1}iv(I-iY)^{-1}
\]%
\[
=i(I+iY)(I-iY)^{-1}((I-iY)(I+iY)^{-1}+I)v(1-iY)^{-1}
\]%
\[
=2iQ(Y)(I+iY)^{-1}v(1-iY)^{-1}.
\]
Thus%
\[
\left\langle dQ_{Y}(v),dQ_{Y}(v)\right\rangle _{Q(Y)}=4\frac{\det(v)}%
{\det(I+Y^{2})}=4\frac{(v,v)_{Y}}{\det(I+Y^{2})}.
\]
Now calculate $\det(I+Y^{2})$ in terms of the $x_{i}$.
\end{proof}

More generally we have

\begin{lemma}
The action of $G$ (or $G_{1}$) on $U(2)$ given by the linear fractional
transformations is conformal relative to the pseudo-Riemannian metric
$\left\langle \ldots,\ldots\right\rangle $ on $U(2)$.
\end{lemma}

\begin{proof}
Let
\[
\phi(Z)=g\cdot Z=(AZ+B)(CZ+D)^{-1}%
\]
with $g\in G_{1}$. Then if $X\in\mathfrak{u}(2)$ we have%
\[
d\phi_{Z}(ZX)=(AZX-\phi(Z)CZX)(CZ+D)^{-1}
\]%
\[
=\phi(Z)(\phi(Z)^{-1}AZX-CZX)(CZ+D)^{-1}.
\]
Thus since $-\det X=\left\langle ZX,ZX\right\rangle _{Z}$ we have
$\left\langle d\phi_{Z}(ZX),d\phi_{Z}(ZX)\right\rangle _{\phi(Z)}=$
\[
 \det\left(  (\phi(Z)^{-1}AZ-CZ)(CZ+D)^{-1}\right)  \left\langle
ZX,ZX\right\rangle _{Z}.
\]
This proves the conformality.
\end{proof}

We note that%
\[
(\phi(Z)^{-1}AZX-CZX)(CZ+D)^{-1}
\]
\[
= (CZ+D)((AZ+B)^{-1}AZ-(CZ+D)^{-1}CZ)(CZ+D)^{-1}.
\]
Thus the conformal factor is%
\[
\det((AZ+B)^{-1}AZ-(CZ+D)^{-1}CZ)
\]%
\[
= \det((AZ+B)^{-1}B-(CZ+D)^{-1}D).
\]

\section{Maxwell's equations on compactified \\ Minkowski space}

We will first recall Maxwell's equations in Lorentzian form. For this we need some
notation. If $M$ is a smooth manifold, then $\Omega^{k}(M)$ will
denote the space of smooth $k$-forms on $M$. We note that if $(M,g)$ is
an $n$-dimensional pseudo-Riemannian manifold then $g$ induces nondegenerate forms
on each fiber $\wedge^{k}T(M)_{x}^{\ast}$ which we will also
denote as $g_{x}$. If $M$ is oriented then there is a unique element
$\gamma\in\Omega^{n}(M)$ such that if $x\in M$ and $v_{1},\ldots,v_{n}$ is an
oriented pseudo-orthonormal basis of $T(M)_{x}$ (i.e., $|g_{x}(v_{i}%
,v_{j})|=\delta_{ij}$) then $\gamma_{x}(v_{1},\ldots,v_{n})=1$. Using $\gamma$
we can define the Hodge $\ast$ operator on $M$ as follows: If $\omega\in \wedge^{k}T(M)_{x}^{\ast}$, then $\ast\omega$ is defined to be the unique
element of $\wedge^{n-k}T(M)_{x}^{\ast}$ such that $\eta\wedge\ast\omega=g_{x}(\eta,\omega)\gamma_{x}$ for all
 $\eta\in\wedge
^{k}T(M)_{x}^{\ast}$.

The next result is standard.

\begin{lemma}
Let $F:M\rightarrow N$ be a conformal, orientation, preserving mapping of
oriented pseudo-Riemannian manifolds. If $\dim M=\dim N=2k$ then $F^{\ast
}\ast\omega=\ast F^{\ast}\omega$ for $\omega\in\Omega^{k}$.
\end{lemma}

With this notation in place we can set up Maxwell's equations. Take $t=x_{4}$
in $\mathbb{R}^{1,3}$ and let $\omega\longmapsto\ast\omega$ denote the Hodge star operator
on differential forms with respect to the Lorentzian structure $(\ldots,\ldots)$ and
the orientation corresponding to $\gamma=dx_{1}\wedge dx_{2}\wedge
dx_{3}\wedge dt$. Then Maxwell's equations in an area free of current (simple
media e.g., in a vacuum, with the dielectric constant, the permeability and
thus the speed of light normalized to $1$) can be expressed in terms of
$2$-forms as%
\begin{equation*}
d\omega=d\ast\omega=0\tag{1}
\end{equation*}
with $d$ the exterior derivative. We note that in this formulation if
$\mathbf{E}=(e_{1},e_{2},e_{3})$ and $\mathbf{H=}(h_{1},h_{2}.h_{3})$ are
respectively the electric field intensity and the magnetic field intensity
vectors, then
\[
\omega=h_{1}dx_{2}\wedge dx_{3}-h_{2}dx_{1}\wedge dx_{3}+h_{3}dx_{1}\wedge
dx_{2}
\]%
\[
-e_{1}dx_{1}\wedge dt-e_{2}dx_{2}\wedge dt-e_{3}dx_{3}\wedge dt\text{.}%
\]
The equations (1) are then the same as%
\[
\nabla\cdot\mathbf{E}=\nabla\cdot\mathbf{H}=0
\]
and%
\[
\frac{\partial}{\partial t}\mathbf{E}=-\nabla\times\mathbf{H},\frac{\partial
}{\partial t}\mathbf{H}=\nabla\times\mathbf{E.}%
\]
Here the $\ast$ operation is just $\mathbf{E}\rightarrow\mathbf{H}$ and
$\mathbf{H}\rightarrow-\mathbf{E}$, which is the duality between electricity
and magnetism\ in the physics literature.

We note that if $M=
\mathbb{R}
^{1,3}$ and $N=U(2)$ with the Lorentzian structures described in the previous
section, if $F$ is the map described above and if $\omega\in
\Omega^{2}(U(2))$ satisfies the equations (1), then $F^{\ast}\omega$ satisfies
Maxwell's equations on $\mathbb{R}^{1,3}$. We will thus call the equations (1) Maxwell's equations on
compactified Minkowski space.

The group $G$ (or $G_{1}$) acts on $U(2)$ by conformal diffeomorphisms. Thus
we see that the space of solutions to Maxwell's equations defines a
representation of $G.$(which acts by pullback). Most of the rest of this
article will be devoted to that analysis of this representation. 

Denote by
$\Omega^{k}(U(2))_{
\mathbb{C}}$ the complex valued $k$-forms. Endow it with the $C^{\infty}$-topology
which is a Fr\'{e}chet space structure and the corresponding action of $G$ on
$\Omega^{2}(U(2))$ defines it as a smooth Fr\'{e}chet representation of $G$
moderate growth. To see this, we note that as a $G$-homogeneous space
$U(2)\cong G/P$. Let $\mu$ denote the isotropy action of $P$ on $V=T_{IP}%
(G/P)_{
\mathbb{C}}$ (i.e., the action of $P$ on $\Lie(G)/\Lie(P)\otimes
\mathbb{C}$). Then the space $\Omega^{k}(U(2))_{
\mathbb{C}}$ with the $C^{\infty}$ topology and $G$ action by pullback is just the
$C^{\infty}$ induced representation%
\[
\rm Ind_{P}^{G}(\wedge^{k}V^{\ast})^{\infty}.
\]
Furthermore, since it is as a $K$-representation%
\[
\rm Ind_{M}^{K}(\wedge^{k}V^{\ast})^{\infty},
\]
Frobenius reciprocity implies that the representation is admissible (that is,
the multiplicities of the $K$-types is finite). The maps $d$ and $d\ast$ are
continuous maps in this topology to $\Omega^{3}(U(2))_{
\mathbb{C}}$; thus the solutions of Maxwell's equations on $U(2)$ define an admissible,
smooth Fr\'{e}chet representation of moderate growth.

\section{The $K$-isotypic components of the space\\ of solutions to Maxwell's
equations on \\compactified Minkowski space: step 1}

In this section we will begin determination of the $K$-isotypic components of
the space of solutions to Maxwell's equations. We will proceed by first
determining the isotypic components of $\ker d$ on $\Omega^{2}(U(2))_{
\mathbb{C}}$. We will then use explicit calculations for the case at hand of $d$ and the
Hodge star operator to complete the picture. We will now begin the first step.

We note that $U(2)$ is diffeomorphic with $SU(3)\times S^{1}$ under the map
\[
u,z\mapsto u\left[
\begin{array}
[c]{cc}%
z & 0\\
0 & 1
\end{array}
\right]
\]
with $u\in SU(2)$ and $z\in S^{1}=\{z\in
\mathbb{C}||z|=1\}$. We note that $SU(2)$ is diffeomorphic with $S^{3}$ which implies
that we have%
\[
H^{1}(U(2),\mathbb{C})=\mathbb{C},\]%
\[H^{2}(U(2),\mathbb{C})=0,
\]
and%
\[
H^{3}(U(2),
\mathbb{C}
)=\mathbb{C}.
\]
So de Rham's theorem implies that we have the following short exact sequences%
\begin{equation*}
0\rightarrow%
\mathbb{C}
1\rightarrow C^{\infty}(U(2),%
\mathbb{C}
)\rightarrow\ker d_{|\Omega^{1}(U(2))_{%
\mathbb{C}
}}\rightarrow%
\mathbb{C}
\mu\rightarrow0;\tag{1}
\end{equation*}
\begin{equation*}
0\rightarrow\ker d_{|\Omega^{1}(U(2))}\rightarrow\Omega^{1}(U(2))_{%
\mathbb{C}
}\rightarrow\ker d_{|\Omega^{2}(U(2))_{\mathbb{C}
}}\rightarrow0\tag{2}
\end{equation*}
in both sequences the map to the kernel is given by $d$. Also $\mu$ is the
image of $\det^{\ast}\frac{dz}{z}$ in the quotient space. These are all
morphisms of smooth Fr\'{e}chet representations of $G$ of moderate growth.

We also note that the center of $G$ consists of the multiples of the identity
and so acts trivially on $U(2)$. Now $K=U(2)\times U(2)$ and the multiples of
the identity correspond under this identification with the diagonal elements
$C=\{(zI,zI)|\left\vert z\right\vert =1\}$ and $M$ is the diagonal $U(2)$ in
$K$. The actual groups acting on $U(2)$ are $K/C$ and $M/C $. We define
$K_{1}=SU(2)\times U(2)$ and $M_{1}=\{(u,u)|u\in SU(2)\}$. Then under the
natural map $K/M=K_{1}/M_{1}$ we still have a redundancy of $\mu_{2}%
=\{\pm(I,I)\}.$ We will use the notation $(\tau_{p,q,r},F^{(p,q,r)})$ for the
representation of $K_{1}$ on $S^{p}(\mathbb{C}
^{2})\otimes S^{q}(\mathbb{C}
^{2})$ ($S^{p}(
\mathbb{C}
^{2})$ the $p^{th}$ symmetric power) given by $\tau_{p,q,r}(u,vz)=z^{r}%
S^{p}(u)\otimes S^{q}(u)$ where $p,q\in
\mathbb{Z}_{\geq0}$, $r\in \mathbb{Z}$ and
 $r\equiv q\operatorname{mod}2$. If $V$ is a closed $K$-invariant
subspace of $\Omega^{k}(U(2))_{\mathbb{C}}$
then we denote by $V_{p,q,r}$ its $\tau_{p,q,r}$ isotypic component.

\begin{lemma}
As a representation of $K_{1}$, the space of $K_{1}$-finite vectors of\linebreak $\ker
d_{|\Omega^{2}(U(2))_{
\mathbb{C}}}$ splits into a direct sum%
\[
{\displaystyle\bigoplus\limits_{%
\begin{array}
[c]{c}%
k\geq0\\
r\equiv k\operatorname{mod}2
\end{array}
}}
(F^{k+2,k,r}\oplus F^{k,k+2,r}\oplus F^{k+1,k+1,r}).
\]
Furthermore, if $p-q\neq0$ then
\[
d:\Omega^{1}(U(2))_{p,q,r}\rightarrow(\ker d_{|\Omega^{2}(U(2))_{\mathbb{C}}})_{p,q,r}
\]
is a bijective $K$-intertwining operator.
\end{lemma}

\begin{proof}
The Peter--Weyl theorem implies that $L^{2}(U(2))$ is a Hilbert space direct
sum
\[
\bigoplus_{\tau\in\widehat{U(2)}}V^{\tau}\otimes(V^{\tau})^{\ast},
\]
where $\widehat{U(2)}$ is the set of equivalence classes of irreducible finite-dimensional 
representations of $U(2)$ and $V^{\tau}$ is a choice of
representative of $\tau$. We have the exact sequence%
\[
1\rightarrow\{\pm(I,1)\}\rightarrow SU(2)\times S^{1}\rightarrow
U(2)\rightarrow1
\]
with the last map $u,z\longmapsto zu$. This implies (as above) that if we
define $V^{p,r}$ to be the representation, $\tau_{p,r}$ of $SU(2)\times S^{1}%
$on $S^{p}(
\mathbb{C}
^{2})$ with $\tau_{p,r}(u,z)v=z^{r}S^{p}(u)v$, then $\tau_{p,r}$ is the lift of
an irreducible representation of $U(2)$ if and only if $r\equiv
p\operatorname{mod}2$. These representations give a complete set of
representatives for $\widehat{U(2)}$. We note that the dual representation of
$\tau_{p,r}$ is equivalent with $\tau_{p,-r}$. We therefore see that the space
of $K$-finite vectors in $C^{\infty}(U(2))_{\mathbb{C}}$ is isomorphic with the direct sum
\[
\bigoplus_{%
\begin{array}
[c]{c}%
p\in
\mathbb{Z}_{\geq0}\\
r\equiv p\operatorname{mod}2
\end{array}
}F^{p,p,r}.
\]
We now apply Frobenius reciprocity to analyze the isotypic components of
$\Omega^{1}(U(2))_{\mathbb{C}}$. 
As we have noted as a representation of $K$ it is just the smooth induced
representation of $M$ to $K$ where
\[
M=\Delta(U(2))=\{(u,u)|u\in U(2)\}
\]
is acting on $\Lie(U(2))_{\mathbb{C}}$ 
under $\Ad(u)$. Thus in terms of the parameters above (identifying $M$ with
$U(2)$) we have
\[
	\Lie(U(2))_{\mathbb{C}}\cong F^{0,0}\oplus F^{2,0}\text{.}%
\]
Now Frobenius reciprocity implies that
\begin{equation}
%\[
\dim {\Hom}_{K}(F^{p,q,r},\Omega^{1}(U(2))_
{\mathbb{C}})=\dim {\Hom}_{M}(F^{p,q,r},F^{0,0})
\end{equation}
\begin{equation}
+\dim {\Hom}_{M}(F^{p,q,r},F^{2,0}).
\end{equation}
%\]
The argument above says that $\dim \Hom_{M}(F^{p,q,r},F^{0,0})=0$ unless $p=q$
and $r\equiv p\operatorname{mod}2$. Now the Clebsch--Gordan formula implies
that%
\[
F_{|M}^{p,q,r}\cong\bigoplus_{j=0}^{\min(p,q)}F^{p+q-2j,r}\text{.}%
\]
This implies that $\dim Hom_{M}(F^{p,q,r},F^{2,0})=0$ unless $p=q$ or
$|p-q|=2$ and in either of these cases it is $1.$Now the exact sequences (1)
and (2) above imply the theorem.
\end{proof}

\begin{remark}\emph{
In the physics literature if $\omega$ is a solution to Maxwell's equations (as
in the beginning of Section 3), then a one-form $\beta$ such that
$d\beta=\omega$ yields in the $\mathbf{E}$, $\mathbf{H}$ formulation a
potential $\mathbf{A}$. In our formulation if we pull back to Minkowski space
and we write}
\[
\beta=\sum_{i=1}^{4}a_{i}dx_{i},%
\]
\emph{then considering $t=x_{4}$ and writing $\phi=a_{4}$ and $\mathbf{A}%
=(a_{1},a_{2},a_{3})$ we then have}
\[
\mathbf{E=\nabla}\times\mathbf{A},\mathbf{H=-}\frac{\partial\mathbf{A}%
}{\partial t}+\nabla\phi.
\]
\emph{This is the dual of what one normally finds in the physics literature. It is
pointed out that this potential has the ambiguity of a gradient field. We will
see that the only isotypic components of Maxwell's equations are $\tau
_{p,q,r}$ with $|p-q|=2$ and $r=\pm(\max(p,q))$. Thus using only those
Peter--Weyl coefficients yields a unique potential.}
\end{remark}

We will use the above lemma and some direct calculations to describe the
$K$-isotypic components of Maxwell's equations in the next section.

\section{The $K$-isotypic components of the space\\ of solutions to Maxwell's
equations on \\compactified Minkowski space: step 2}

We retain the notation of the previous section. Let $x_{4}=iI$ and%
\[
x_{1}\,=\,\left[
\begin{array}
[c]{cc}%
i & 0\\
0 & -i
\end{array}
\right]  ,x_{2}=\left[
\begin{array}
[c]{cc}%
0 & 1\\
-1 & 0
\end{array}
\right]  ,x_{3}=\left[
\begin{array}
[c]{cc}%
0 & i\\
i & 0
\end{array}
\right]  .
\]

We will use the usual identification of $\Lie(U(2))$(left invariant vector
fields) with skew-Hermitian $2\times2$ matrices (which we denote, as is usual, by
$\mathfrak{u}(2)$). Thus if $x\in\mathfrak{u}(2)$ then $x_{u}$ is the tangent
vector at $0$ to the curve $t\longmapsto ue^{tx}$. We note that $((x_{j}%
)_{u},(x_{k})_{u})_{u}=\varepsilon_{j}\delta_{j,k}$ with $\varepsilon
_{j}=-(-1)^{\delta_{j4}}.$ Thus $x_{1},x_{2},x_{3},x_{4}$ define a
pseudo-orthonormal frame on $U(2)$. We use this frame to define $\gamma$.
Since there will be many uses of the star operator and pullbacks we will use
the notation $J\omega=\ast\omega$ for $\omega\in\wedge^{2}T^{\ast}(U(2))_{u}$
for all $u\in U(2)$. We define $\alpha_{j}$ to be the left invariant one-form
on $U(2)$ defined by $\alpha_{j}(x_{k})=\delta_{jk}$. We note that
\[
J\alpha_{1}\wedge\alpha_{2}=\alpha_{3}\wedge\alpha_{4},J\alpha_{1}\wedge
\alpha_{3}=-\alpha_{2}\wedge\alpha_{4},J\alpha_{2}\wedge\alpha_{3}=\alpha
_{1}\wedge\alpha_{4}%
\]
and%
\[
J\alpha_{1}\wedge\alpha_{4}=-\alpha_{2}\wedge\alpha_{3},J\alpha_{2}%
\wedge\alpha_{4}=\alpha_{1}\wedge\alpha_{3},J\alpha_{3}\wedge\alpha
_{4}=-\alpha_{1}\wedge\alpha_{2}.
\]
From this we note

\begin{lemma}
We have $J^{2}=-I$ on each space $\wedge^{2}T^{\ast}(U(2))_{u}$. Furthermore a
basis of the eigenspace for $i$ in
$
\wedge^{2}T^{\ast}(U(2))_{u}\otimes
\mathbb{C}
$
is 
\[
\mathcal{B}_{i}=\{\alpha_{1}\wedge\alpha_{4}+i\alpha_{2}\wedge\alpha
_{3},\alpha_{2}\wedge\alpha_{4}-i\alpha_{1}\wedge\alpha_{3},\alpha_{3}
\wedge\alpha_{4}+i\alpha_{1}\wedge\alpha_{2}\}
\]
 a basis for the eigenspace
$-i$ is 
\[
\mathcal{B}_{-i}=\{\alpha_{1}\wedge\alpha_{4}-i\alpha_{2}\wedge
\alpha_{3},\alpha_{2}\wedge\alpha_{4}+i\alpha_{1}\wedge\alpha_{3},\alpha
_{3}\wedge\alpha_{4}-i\alpha_{1}\wedge\alpha_{2}\}.
\]
If $\mu\in\mathcal{B}%
_{i}$ and $\nu\in\mathcal{B}_{-i}$,  then $\mu\wedge\nu=0$.
\end{lemma}

We look upon $J$ as an operator on $\Omega^{2}(U(2))_{
\mathbb{C}}$. Since $J$ preserves the real vector space $\Omega^{2}(U(2))$ we get a
decomposition%
\[
\Omega^{2}(U(2))_{
\mathbb{C}}=\Omega^{2}(U(2))_{i}\oplus\Omega^{2}(U(2))_{-i}%
\]
with $J_{|\Omega^{2}(U(2))_{\pm i}}=\pm iI$. If $\omega\in\Omega^{2}(U(2))_{
\mathbb{C}}$ then we denote by $\overline{\omega}$ the complex conjugate of $\omega$
relative to the real space $\Omega^{2}(U(2))$. We note

\begin{lemma} With the notation above, we have
\[
\Omega^{2}(U(2))_{\pm i}=\bigoplus_{\mu\in\mathcal{B}_{\pm i}}C^{\infty}(U(2),
\mathbb{C}
)\mu.
\] 
\end{lemma}

We now calculate the exterior derivatives of the $\alpha_{j}$. We observe that
if $\alpha$ is a left invariant element of $\Omega^{1}(U(2))$, then%
\[
d\alpha(x,y)=-\alpha([x,y])
\]
for $x,y\in \Lie(U(2))$. This implies that%

\begin{equation*}
d\alpha_{1}=-2\alpha_{2}\wedge\alpha_{3},d\alpha_{2}=2\alpha_{1}\wedge
\alpha_{3},d\alpha_{3}=-2\alpha_{1}\wedge\alpha_{2}. \tag{1}%
\end{equation*}
We also note that $d\alpha_{4}=0$.

We denote by $\chi_{k}$ the character of $SU(2)\times S^{1}$ given by
$\chi_{k}(u,z)=z^{k}$. We denote by $\pi$ the covering map $\pi:SU(2)\times
S^{1}\rightarrow U(2)$ given by $\pi(u,z)=uz$. Then we have

\begin{lemma} We have
\[
(\ker d_{|\Omega^{2}(U(2))_{\mathbb{C}}
})_{k,k,r} = \chi_{r}dC^{\infty}(U(2),
\mathbb{C}
)_{k,k,0}\wedge\alpha_{4}
\]
which is defined on $U(2)$ if $r\equiv
k\operatorname{mod}2$.
\end{lemma}

\begin{proof}
We note that if $f\in C^{\infty}(U(2),\mathbb{C})$, 
then $df=(x_{4}f)\alpha_{4}+\nu$ with $\nu=\sum_{j}(x_{j}f)\alpha_{j}. $
Thus if $df\wedge\alpha_{4}=0$ then $\nu=0$. If $\nu=0$ then $\pi^{\ast
}f(u,z)=\pi^{\ast}f(I,z)$. A function in $C^{\infty}(U(2),\mathbb{C})_{k,k,0}$
with this property exists if and only if $k=0$. It is also clear that
$\chi_{r}dC^{\infty}(U(2),
\mathbb{C}
)_{k,k,0}\wedge\alpha_{4}$ is contained in $\ker d$. Thus, since each of the
isotypic components of $\ker d_{|\Omega^{2}(U(2))_{\mathbb{C}
}}$ is irreducible and we have accounted for all of them by Lemma 4, the result follows.
\end{proof}

\bigskip

We denote by Maxw the space of complex solutions to the Maxwell equations (as
described in the previous section). Then Maxw is a closed subspace of
$\Omega^{2}(U(2))_{\mathbb{C}}$ 
yielding a smooth Fr\'{e}chet representation of $G$ of moderate growth
under the action $\mu(g)\omega=(g^{-1})^{\ast}\omega$. We also note that $J$
preserves Maxw and commutes with the action of $G$. This implies that Maxw $=$
Maxw$_{i}$ $\oplus$ Maxw$_{-i}$ (corresponding to the $i$ and $-i$ eigenspaces
of $J$ on Maxw). We also note

\begin{lemma}
\rm{Maxw}$_{\pm i}=\{\omega\in\Omega^{2}(U(2))_{
\mathbb{C}
}|J\omega=\pm i\omega$ and $d\omega=0\}$.
\end{lemma}

We can now eliminate some isotypic components of Maxw.

\begin{lemma}
{\rm{Maxw}}$_{k,k,r}=0$ for all $k\in\mathbb{Z}
_{\geq0}$ and all $r\in\mathbb{Z}$.
\end{lemma}

\begin{proof}
If $\omega\in$ {\rm{Maxw}}$_{k,k,r}$, then since 
$$
{\rm{Maxw}}_{k,k,r}= {\rm{Maxw}}_{k,k,r}\cap
{\rm{Maxw}}_{i}\oplus{\rm{ Maxw}}_{k,k,r}\cap{\rm{ Maxw}}_{-i}
$$
 with each of the summands
$K$-invariant and since each isotypic component is irreducible we see that
$J\omega=i\omega$ of $J\omega=-i\omega$. In either case Lemma 6 implies that
$\omega$ is not an element of $\Omega^{1}(U(2))_{
\mathbb{C}}\wedge\alpha_{4}$. But Lemma 7 implies that it must be of that form. This
implies that $\omega=0$.
\end{proof}

\bigskip

We are finally ready to give the isotypic components of Maxw.

\begin{theorem}
{\rm{Maxw}}$_{p,q,r}$ is nonzero if and only if $(p,q,r)$ is in one the following forms
\[
(k+2,k,k+2),(k+2,k,-(k+2)),(k,k+2,k+2),(k,k+2,-(k+2)).
\]
If it is nonzero it is irreducible. Moreover, {\rm{ Maxw}}$_{k+2,k,k+2}$ and {\rm{Maxw}}$_{k,k+2,-k-2}$ are contained in
{\rm{Maxw}}$_{i}$ and {\rm{Maxw}}$_{k+2,k,-k-2}$ and {\rm{Maxw}}$_{k,k+2,k+2}$ are contained in
{\rm{Maxw}}$_{-i}$.
\end{theorem}

The proof will occupy the rest of the section. If $x\in\mathfrak{s}l(2,\mathbb{C})$ 
with $x=v+iw,v,w\in\mathfrak{s}u(2)$, then we define the left invariant
vector field $x^{L}f(u)=\frac{d}{dt}(f(ue^{tv})+if(ue^{tw}))_{|t=0}$ and
$x^{R}f(u)=\frac{d}{dt}(f(e^{-tv}u)+if(e^{-tw}v))_{|t=0}$. We will think of
these vector fields as being on $SU(2)$ or $SU(2)/\{\pm I\}=U(2)/S^{1}%
I=PSU(2)$. We set $e=\frac{1}{2}(x_{2}-ix_{3})$, $f=-\frac{1}{2}(x_{2}%
+ix_{3})$ and $h=-ix_{1}.$ Then $e,f,h$ form the standard basis of
$\mathfrak{s}l(2,\mathbb{C})$.

Let $\xi_{e}^{L},\xi_{f}^{L},\xi_{h}^{L}$ (respectively, $\xi_{e}^{R},\xi
_{f}^{R},\xi_{h}^{R}$) be the left (resp. right) invariant one-forms  on
$PSU(2)$ that form a dual basis to $e^{L},f^{L},h^{L}$ (respec\-tively, $e^{R},f^{R},h^{R}$). Let $p:U(2)\rightarrow PSU(2)$ be the obvious quotient
homomorphism and let $\alpha_{b}^{a}=p^{\ast}\xi_{b}^{a}$ for $a=L$ or $R$ and
$b=e,f$ or $h$. Now $\,((u,v)^{-1})^{\ast}\alpha_{b}^{L}(x^{L})=\alpha_{b}%
^{L}(\Ad(v)^{-1}x^{L})$ and $((u,v)^{-1})^{\ast}\alpha_{b}^{R}(x^{R}%
)=\alpha_{b}^{R}(\Ad(u)^{-1}x^{R})$. Thus, if
\[
g=\left(  \left[
\begin{array}
[c]{cc}%
z & 0\\
0 & z^{-1}%
\end{array}
\right]  ,\left[
\begin{array}
[c]{cc}%
w & 0\\
0 & w^{-1}%
\end{array}
\right]  \right)
\]
with $z,w\in S^{1}$, then
\[
(g^{-1})^{\ast}\alpha_{e}^{L}=w^{-2}\alpha_{e}^{L},(g^{-1})^{\ast}\alpha
_{f}^{L}=w^{2}\alpha_{f}^{L},(g^{-1})^{\ast}\alpha_{h}^{L}=\alpha_{h}^{L}%
\]
and%
\[
(g^{-1})^{\ast}\alpha_{e}^{R}=u^{-2}\alpha_{e}^{R},(g^{-1})^{\ast}\alpha
_{f}^{R}=u^{2}\alpha_{f}^{R},(g^{-1})^{\ast}\alpha_{h}^{R}=\alpha_{h}^{R}.
\]
This implies that%
$$
{\rm{Span}}_{\mathbb{C}
}(\alpha_{e}^{L},\alpha_{f}^{L},\alpha_{h}^{L})=\Omega^{1}(U(2))_{0,2,0}%
$$
and
$$
{\rm Span}_{\mathbb{C}
}(\alpha_{e}^{R},\alpha_{f}^{R},\alpha_{h}^{R})=\Omega^{1}(U(2))_{2,0,0}.
$$
Also relative to the positive root system $g\rightarrow u^{2},g\rightarrow
w^{2}$ the highest weight space of $\Omega^{1}(U(2))_{0,2,0}$ in 
$\mathbb{C}\alpha_{f}^{L}$ and that of $\Omega^{1}(U(2))_{2,0,0}$ is 
$\mathbb{C}\alpha_{f}^{R}$.

We can now describe highest weight vectors for the isotypic components
$\Omega^{1}(U(2))_{k,k+2,l}$ and $\Omega^{1}(U(2))_{k+2,k,l}$. Let $e_{1}$ and
$e_{2}$ be the standard basis of 
$\mathbb{C}^{2}$. Fix a $U(2)$ invariant inner product 
$\ \left\langle\ldots,\ldots\right\rangle $
 on each space $S^{k}(\mathbb{C}^{2})$. Define%
\[
\phi_{k}(u)=\left\langle S^{k}(u)e_{1}^{k},e_{2}^{k}\right\rangle .
\]
Then $\phi_{k}$ is a highest weight vector for $C^{\infty}(U(2))_{k,k,k}%
=\Omega^{0}(U(2))_{k,k,k}$. Also we define $\chi_{l}(u,z)=z^{l}$ for $u\in
SU(2)$ and $z\in S^{1}$. We note that if $l\equiv k\operatorname{mod}2$, then
\[
\psi_{k,l}(uz)=z^{l-k}\phi_{k}(uz)=\chi_{l-k}(u,z)\phi_{k}(uz)
\]
is defined and is a highest weight vector for $\Omega^{0}(U(2))_{k,k,l}$. This
implies that $\psi_{k,l}\alpha_{f}^{L}$ is a highest weight vector for
$\Omega^{1}(U(2))_{k,k+2,l}$ and $\psi_{k,l}\alpha_{f}^{R}$ is a highest
weight vector for $\Omega^{1}(U(2))_{k+2,k,l}$. We have shown that Maxw is
multiplicity free and Maxw $=$ Maxw$_{i}$ $\oplus$ Maxw$_{-i}$. We have also
proved that Maxw$_{k,k+2,l}=d\Omega^{1}(U(2))_{k,k+2,l}$ and Maxw$_{k+2,k,l}%
=d\Omega^{1}(U(2))_{k+2,k,l}$. We have proved that

\medskip

\noindent (1) Maxw$_{k,k+2,l}\neq0$ if and only if $Jd\psi_{k,l}\alpha_{f}%
^{L}=\lambda J\psi_{k,l}\alpha_{f}^{L}$ with $\lambda\in\{\pm i\}$ and
Maxw$_{k+2,k,l}\neq0$ if and only if $Jd\psi_{k,l}\alpha_{f}^{R}=\lambda
J\psi_{k,l}\alpha_{f}^{R}$ with $\lambda\in\{\pm i\}$.

\medskip

We are now left with a computation. One checks as above (using $d\alpha
_{f}^{L}(X^{L},Y^{L})=-\alpha_{f}^{L}([X,Y]^{L})$)
\[
d\alpha_{f}^{L}=2\alpha_{h}^{L}\wedge\alpha_{f}^{L};%
\]
we therefore have%
\[
d(\psi_{k,l}\alpha_{f}^{L})=il\psi_{k,l}\alpha_{4}\wedge\alpha_{f}^{L}%
+(k+2)\psi_{k,l}\alpha_{h}^{L}\wedge\alpha_{f}^{L}.
\]

We consider $il\alpha_{4}\wedge\alpha_{f}^{L}+(k+2)\alpha_{h}^{L}\wedge
\alpha_{f}^{L}$. We observe that $\alpha_{h}^{L}=i\alpha_{1}$ and $\alpha
_{f}^{L}=-(\alpha_{2}-i\alpha_{3})$. Thus (using the calculations leading to
Lemma 6) the right-hand side of the equation above is equal to

\[
li\alpha_{4}\wedge\alpha_{2}-l\alpha_{4}\wedge\alpha_{3}+i(k+2)\alpha
_{1}\wedge(\alpha_{2}-i\alpha_{3})
\]
which equals
\[
(-li\alpha_{2}\wedge\alpha_{4}+(k+2)J\alpha_{2}\wedge\alpha_{4})+(-l\alpha
_{1}\wedge\alpha_{4}-i(k+2)J\alpha_{1}\wedge\alpha_{4}.
\]
We therefore see that if $l>0$ then this expression is an element of $\Omega
^{2}(U(2))_{-i}+(l-k-2)(-i\alpha_{2}\wedge\alpha_{4}-(l-k-2)\alpha_{1}%
\wedge\alpha_{4})$. If $l\leq0$ then it is an element of $\Omega^{2}(U(2))_{i}%
+(l+k+2)(-i\alpha_{2}\wedge\alpha_{4}-(l+k+2)\alpha_{1}\wedge\alpha_{4})$. Thus we have

\medskip

\noindent (2) Maxw$_{k,k+2,l}\neq0$ only if $l=k+2$ or $l=-(k+2)$. Furthermore,
Maxw$_{k,k+2,k+2}\subset$ Maxw$_{-i}$ and Maxw$_{k,k+2,-(k+2)}\subset$
Maxw$_{i}$.

\medskip

We note that everything that we have done could have been done with right-invariant vector fields to complete the proof of the theorem$.$ However we
will proceed in a different way. Let $\eta:U(2)\rightarrow U(2)$ be defined by
$\eta(u)=u^{-1}$. Then for $x\in\mathfrak{u}(2)$, we have $d\eta_{u}(x_{u}%
^{L})=x_{u^{-1}}^{R}$. This implies that
\[
\left\langle d\eta_{u}(x_{u}^{L}),d\eta_{u}(x_{u}^{L})\right\rangle _{u^{-1}%
}=\left\langle x_{u^{-1}}^{R},x_{u^{-1}}^{R}\right\rangle _{u^{-1}}=-\det x
\]
since%
\[
\left\langle x_{u^{-1}}^{R},x_{u^{-1}}^{R}\right\rangle _{u^{-1}}=\left\langle
dR(u^{-1})_{I}(x_{I}^{R}),dR(u^{-1})_{I}(x_{I}^{R})\right\rangle _{u^{-1}%
}=\left\langle x_{I}^{R},x_{I}^{R}\right\rangle _{I}.
\]
This proves that $\eta$ is an isometry. It also implies that
\[
(\eta^{\ast}\alpha_{f}^{R})_{u}(x_{u}^{L})=(\alpha_{f}^{R})_{u^{-1}}(d\eta
_{u}(x_{u}^{L}))=(\alpha_{f}^{R})_{u^{-1}}(x_{u^{-1}}^{R}).
\]
Hence $\eta^{\ast}\alpha_{f}^{R}=\alpha_{f}^{L}$. 
Now $\eta^{\ast}d(\psi_{k,l}\alpha_{f}^{R})\in\Omega^{2}(U(2))_{k,k+2,-l}$ since $\eta^{\ast
}\phi_{k}$ is a highest weight vector for $C^{\infty}(U(2))_{k,k,-k} $ . We
also note that $\eta^{\ast}\gamma=\gamma$. Thus $\eta^{\ast}$Maxw $=$ Maxw.
Hence, if $d(\psi_{k,l}\alpha_{f}^{R})\in$ Maxw$_{k+2,k,l}$,  then $\eta^{\ast
}d(\psi_{k,l}\alpha_{f}^{R})\in$ Maxw$_{k,k+2,-l}$. So, if $l>0$ then, we must
have $l=k+2$ and if $l\leq0$, then $l=-k-2$. Since $\eta^{\ast}$ commutes with
$J$ the last assertion also follows.

\begin{remark}\emph{We have }
\begin{equation} 
\eta^{\ast}{\rm{Maxw}}_{k+2,k,k+2}={\rm{Maxw}}_{k,k+2,-k-2},
\end{equation}
\begin{equation}
\eta^{\ast}{\rm{Maxw}}_{k+2,k,-k-2}= {\rm{Maxw}}_{k,k+2,k+2}
 {\emph{ and }} (\eta^{\ast})^{2}=I.
\end{equation}
\end{remark}

\section{The Hermitian form}

We retain the notation of the previous sections.

We note that%
\[
H_{3}(U(2),\mathbb{R})=\mathbb{R}.
\]
The form $\nu=\alpha_{1}\wedge\alpha_{2}\wedge\alpha_{3}$ restricted to
$SU(2)$ satisfies%
\[
\int_{SU(2)}\nu=2\pi^{2}
\]
and $d\nu=0$. This implies that the class of $SU(2)$ in the third homology
over $\mathbb{R}$ is a basis. (In fact it is well known that this is true over 
$\mathbb{Z}$). We note that this implies, in particular, that if $\omega\in\Omega
^{3}(U(3))_{\mathbb{C}}$ satisfies $d\omega=0$ then if $M$ is a compact submanifold such that there
exists a smooth family of diffeomorphisms of $U(2)$, $\Phi_{t}$ such that
$\Phi_{0}=I$ and $\Phi_{1}(SU(2))=M$, then%
\[
\int_{SU(2)}\omega=\int_{M}\omega.
\]
In particular we have (since U(2,2) is connected),

\begin{lemma}
If $g\in U(2,2)$ and if $\omega\in\Omega^{3}(U(3))_{\mathbb{C}}$ 
satisfies $d\omega=0$, then $\int_{SU(2)}\omega=\int_{gSU(2)}\omega
=\int_{SU(2)}g^{\ast}\omega.$
\end{lemma}

We will apply this observation to define $U(2,2)$ invariant sesquilinear forms
on the spaces {\rm{Maxw}}$_{\pm i}$.

\begin{lemma}
Let $\alpha\in\Omega^{1}(U(2))_{\mathbb{C}}$ 
and $\omega\in\Omega^{2}(U(2))_{\mathbb{C}}$ 
be such that $\omega,d\alpha\in \rm{Maxw}_{i}$ (resp.~{\rm{ Maxw}}$_{-i}$). Then
$d(\alpha\wedge\overline{\omega})=0$.
\end{lemma}

\begin{proof}
If $d\alpha\in$ Maxw$_{i}$ then $\overline{\omega}\in {\rm{Maxw}}_{-i}$. We note
that Maxwell's equations imply that $d\overline{\omega}=0$. Hence
$d(\alpha\wedge\overline{\omega})=d\alpha\wedge\overline{\omega}$ and Lemma 6
implies that \ Maxw$_{i}\wedge$ Maxw$_{-i}=0$. Obviously the same argument
works for $-i$.
\end{proof}

\begin{proposition}
If $\mu,\omega\in {\rm{Maxw}}_{i}$ (or  \rm{Maxw}$_{-i}$), there exists $\alpha
\in\Omega^{1}(U(2))_{\mathbb{C}}$ 
such that $d\alpha=\omega$. The expression%
\[
\int_{SU(2)}\alpha\wedge\overline{\mu}%
\]
depends only on $\omega$ and $\mu$ (and not on the choice of $\alpha$).
Furthermore, the integral defines a Hermitian form $\left\langle \omega
,\mu\right\rangle $ on  \rm{Maxw}$_{i}$ (or \rm{Maxw}$_{-i}$) that satisfies
$\left\langle g^{\ast}\omega,g^{\ast}\mu\right\rangle =\left\langle \omega
,\mu\right\rangle $ for all $g\in U(2,2)$.
\end{proposition}

\begin{proof}
Suppose that $\beta\in\Omega^{1}(U(2))_{\mathbb{C}}$ 
is such that $d\beta=0$ and $\int_{S^{1}I}\beta=0$. Then since $H_{1}(U(2),\mathbb{C})$ 
is spanned by the class of $S^{1}I$, de Rham's theorem implies that there
exists $f\in C^{\infty}(U(2),\mathbb{C})$ 
such that $df=\beta$. Set $\iota(\beta)=\int_{S^{1}I}\beta$. We note that%
\[
\iota(\alpha_{1})=2\pi.
\]
We also note that if $\nu\in\Omega^{2}(U(2))$ then $(\alpha_{1}\wedge
\nu)_{|SU(2)}=0.$

We  observed that if $\omega\in$ Maxw$_{\pm i}$, then there exists
$\alpha\in\Omega^{1}(U(2))_{\mathbb{C}}$ 
such that $d\alpha=\omega$. If $d\beta=\omega$ then $d(\beta-\alpha)=0$ and
$\int_{S^{1}I}(\beta-\alpha-\frac{\iota(\beta-\alpha)}{2\pi}\alpha_{1})=0$ so
$\beta-\alpha-\frac{\iota(\beta-\alpha)}{2\pi}\alpha_{1}=-df$ with $f\in
C^{\infty}(U(2),\mathbb{C})$. 
This implies that
\[
\int_{SU(2)}\alpha\wedge\overline{\mu}-\int_{SU(2)}\beta\wedge\overline{\mu
}=\int_{SU(2)}(df-\frac{\iota(\beta-\alpha)}{2\pi}\alpha_{1})\wedge
\overline{\mu}
\]%
\[=
\int_{SU(2)}d(f\overline{\mu})-\frac{\iota(\beta-\alpha)}{2\pi}\int%
_{SU(2)}\alpha_{1}\wedge\overline{\mu}=0.
\]
Both of the integrals are $0$. We next observe that $\mu=d\xi$. We have%
\[
\overline{\left\langle \mu,\omega\right\rangle }-\left\langle \omega
,\mu\right\rangle =\int_{SU(2)}\left(  \overline{\xi}\wedge d\alpha
-\alpha\wedge d\overline{\xi}\right)\hskip72pt 
\]%
\[=
\int_{SU(2)}\left(  \overline{\xi}\wedge d\alpha-d\overline{\xi}\wedge
\alpha\right)  =-\int_{SU(2)}d\left(  \overline{\xi}\wedge\alpha\right)  =0.
\]
We are left with the proof of $U(2,2)$ invariance. We will concentrate on
 \rm{Maxw}$_{i}$. We note that
\[
\int_{S^{1}I}\alpha_{1}=2\pi\text{.}%
\]
Let $g\in U(2,2)$. If $\omega,\mu\in {\rm{Maxw}}_{i}$ and $\alpha\in\Omega
^{1}(U(2))_{\mathbb{C}}$ 
satisfies $d\alpha=\omega$, then $dg^{\ast}\alpha=g^{\ast}\omega.$ Thus
\[
\left\langle g^{\ast}\omega,g^{\ast}\mu\right\rangle =\int_{SU(2)}g^{\ast
}\alpha\wedge g^{\ast}\overline{\mu}=\int_{SU(2)}\alpha\wedge\overline{\mu
}=\left\langle \omega,\mu\right\rangle
\]
by Lemma 13.
\end{proof}

\bigskip
We will now calculate $\left\langle \ldots,\ldots\right\rangle $ on each of the
isotypic components of Maxw. We set $\alpha_{k,k+2,l}=\psi_{k,l}\alpha_{f}
^{L}$ as in the previous section. If $l=k+2$ or $-(k+2)$, then $\omega
_{k,k+2,l}=d\alpha_{k,k+2,l}$ is a highest weight vector for \rm{Maxw}$_{k,k+2,l} $.

\begin{lemma}
If $l=\pm(k+2)$ then $\left\langle \omega_{k,k+2,l},\omega_{k,k+2,l}%
\right\rangle =-\frac{4k+8}{k+1}\pi^{2}$.
\end{lemma}

\begin{proof}
We have seen in formula (2) in the previous section that (in the notation
therein)%
\[
\omega_{k,k+2,l}=il\psi_{k,l}\alpha_{4}\wedge\alpha_{f}^{L}+(k+2)\psi
_{k,l}\alpha_{h}^{L}\wedge\alpha_{f}^{L}.
\]
Thus since $\alpha_{4|SU(2)}=0$ we have%
\[
\alpha_{k,k+2,l}\wedge\overline{\omega_{k,k+2,l}}_{|SU(2)}=(k+2)|\psi
_{k,l}|^{2}\alpha_{f}^{L}\wedge\overline{\alpha_{h}^{L}\wedge\alpha_{f}^{L}}.
\]
Using the expressions for $\alpha_{h}^{L}$ and $\alpha_{f}^{L}$ one sees that%
\[
\alpha_{f}^{L}\wedge\overline{\alpha_{h}^{L}\wedge\alpha_{f}^{L}}%
_{|SU(2)}=-2\alpha_{1}\wedge\alpha_{2}\wedge\alpha_{3|SU(2)}.
\]
Normalized invariant measure on SU(2) is $\mu=\frac{1}{2\pi^{2}}\alpha
_{1}\wedge\alpha_{2}\wedge\alpha_{3|SU(2)}$. On $S^{k}(\mathbb{C}^{2})$ 
we put the tensor--product  inner product. Thus $e_{1}^{k}$ and
$e_{2}^{k}$ are unit vectors. Also%
\[
|\psi_{k,l}(u)|^{2}=\left\langle S^{k}(u)e_{1}^{k},e_{2}^{k}\right\rangle
\overline{\left\langle S^{k}(u)e_{1}^{k},e_{2}^{k}\right\rangle}.%
\]
The Schur orthogonality relations imply that%
\[
\int_{SU(2)}|\psi_{k,l}(u)|^{2}\mu=\frac{1}{k+1}. 
\]
\vbox{\hskip338pt}
 \end{proof}

\section{Four unitary ladder representations\\ of $SU(2,2)$}

We will consider $U(2,2)$ in the usual form that is if $I_{2,2}=\left[
\begin{array}
[c]{cc}%
I & 0\\
0 & -I
\end{array}
\right]  $, then (as in Section 2)
\[
G_{1}=\left\{  g\in SL(4,\mathbb{C})|gI_{2,2}g^{\ast}=I_{2,2}\right\}  .
\]
In this form $K$ is the subgroup of block diagonal matrices. $\mathfrak{g}%
=\Lie(U(2,2))_{\mathbb{C}}=M_{4}(\mathbb{C})$. 
We set%
\[
\mathfrak{p}^{+}=\left\{  \left[
\begin{array}
[c]{cc}%
0 & X\\
0 & 0
\end{array}
\right]  |X\in M_{2}(\mathbb{C})\right\}
\]
and%
\[
\mathfrak{p}^{-}=\left\{  \left[
\begin{array}
[c]{cc}%
0 & 0\\
Y & 0
\end{array}
\right]  |Y\in M_{2}(\mathbb{C})\right\}  .
\]
Then $\mathfrak{g}=\Lie(K)_{\mathbb{C}
}\oplus\mathfrak{p}$ (here $\mathfrak{p}$ is the orthogonal complement of
$\Lie(K)_{\mathbb{C}}$ 
relative to the trace form) \ and as a $K$-module $\mathfrak{p=p}^{+}%
\oplus\mathfrak{p}^{-}$. We leave it to the reader to check

\begin{lemma}
As a representation of $K$ (under the adjoint action) $\mathfrak{p}^{+}\cong
F^{1,1,-1}$ and $\mathfrak{p}^{-}\cong F^{1,1,1}$ and%
\[
\wedge^{2}(\mathfrak{p}^{+})^{\ast}\cong F^{2,0,2}\oplus F^{0,2,2}%
\]
and%
\[
\wedge^{2}(\mathfrak{p}^{-})^{\ast}\cong F^{2,0,-2}\oplus F^{0,2,-2}.
\]

\end{lemma}

Let $\pi(g)\omega=(g^{-1})^{\ast}\omega$ for $\omega\in\rm{Maxw}$. Then we already observed that with the $C^{\infty}$-topology,  $(\pi, \rm{Maxw})$ is an
admissible smooth Fr\'{e}chet representation of moderate growth. We set
Maxw$_{K}$ equal to the space of $K$-finite vectors of $\rm{Maxw}$   and we will use module
notation for the action of $\mathfrak{g}=\Lie(U(2,2))_
{\mathbb{C}}=M_{4}(\mathbb{C})$
 on  {\rm{Maxw}}$_{K}$, thereby having an admissible finitely generated
$(\mathfrak{g},K)$-module.

\begin{lemma}
$\mathfrak{p}^{+}$ annihilates {\rm{Maxw}}$_{2,0,2}$ and {\rm{Maxw}}$_{0,2,2}$,
$\mathfrak{p}^{-}$annihilates  $\rm{Maxw}_{2,0,-2}$ and  $\rm{Maxw}_{0,2,-2}.$
\end{lemma}

\begin{proof}
Using the Clebsch--Gordon formula we have%
\[
F^{1,1,-1}\otimes F^{2,0,2}\cong F^{3,1,1}\oplus F^{1,1,1},
\]%
\[
F^{1,1,-1}\otimes F^{0,2,2}\cong F^{1,3,1}\oplus F^{1,1,1},
\]%
\[
F^{1,1,1}\otimes F^{2,0,-2}\cong F^{3,1,-1}\oplus F^{1,1,-1},
\]%
\[
F^{1,1,1}\otimes F^{0,2,-2}\cong F^{1,3,-1}\oplus F^{1,1,-1}.
\]
Theorem 11 implies that none of the $K$-types on the right of these equations
occurs in Maxw.
\end{proof}

\bigskip

We set
$$ ({\rm Maxw}_{2,0}^{+})_{K}=\sum_{k\geq0}{\rm Maxw}_{k+2,k,k+2},\,
 ({\rm Maxw}_{0,2}^{+})_{K}=\sum_{k\geq0}{\rm Maxw}_{k,k+2,k+2},$$
$$({\rm Maxw}_{2,0}^{-})_{K}=\sum_{k\geq 0}  {\rm Maxw}_{k+2,k,-k-2} \, {\rm and}
({\rm Maxw}_{0,2}^{-})_{K}
=\sum_{k\geq0}{\rm Maxw}_{k,k+2,-k-2}.$$
 We will drop the sub-$K$ for the
completions of these spaces. We note Theorem 11 implies that each of these
spaces is totally contained in Maxw$_{i}$ or Maxw$_{-i}$. This implies that
the Hermitian form, $\left\langle\ldots,\ldots\right\rangle $ from the previous
section is defined on each of these spaces.

We have
\begin{theorem}
Each of the spaces  $\rm{Maxw}_{2,0}^{+}$, $\rm{Maxw}_{2,0}^{-}$, $\rm{Maxw}_{0,2}^{+}$,
 $\rm{Maxw}_{0,2}^{-}$ is an invariant irreducible subspace for $\pi$. Furthermore,
the form $\left\langle \ldots,\ldots\right\rangle $ is positive definite on
 $\rm{Maxw}_{2,0}^{+}$,  $\rm{Maxw}_{2,0}^{-}$ and negative definite on  $\rm{Maxw}_{0.2}^{+}$,
 $\rm{Maxw}_{0,2}^{-}$.
\end{theorem}

\begin{proof}
We consider  $\rm{Maxw}_{2,0}^{+}$. We note that using the Clebsch--Gordan formula as
above we have%
\[
\mathfrak{p}^{-}\text{Maxw}_{k+2,k,k+2}\subset\text{Maxw}_{k+3,k+1,k+3}%
+ \,\text{Maxw}_{k+3,k-1,k+3}\]
\[+\text{Maxw}_{k+1,k+2,k+3}+\text{Maxw}_{k+1,k-1,k+3}%
\]
and%
\[
\mathfrak{p}^{+}\text{Maxw}_{k+2,k,k+2}\subset\text{Maxw}_{k+3,k+1,k+1}%
+ \, \text{Maxw}_{k+3,k-1,k+1}\]
\[+\text{Maxw}_{k+1,k+1,k+1}+\text{Maxw}%
_{k+1,k-1,k+1}.
\]
Using the form of the $K$-types of $\rm{Maxw}$  we see that%
\[
\mathfrak{p}^{-}\text{Maxw}_{k+2,k,k+2}\subset\text{Maxw}_{k+3,k+1,k+3}%
\]
and%
\[
\mathfrak{p}^{+}\text{Maxw}_{k+2,k,k+2}\subset\text{Maxw}_{k+1,k-1,k+1}.
\]
This implies that $( \rm{Maxw}_{2,0}^{+})_{K}$ is a $(\mathfrak{g,}K)$-submodule
of  $\rm{Maxw}$. Hence  $\rm{Maxw}_{2,0}^{+}$ is $G$-invariant. The same argument proves%
\[
\mathfrak{p}^{-}\text{Maxw}_{k,k+2,k+2}\subset\text{Maxw}_{k+1,k+3,k+3},
\]%
\[
\mathfrak{p}^{+}\text{Maxw}_{k,k+2,k+2}\subset\text{Maxw}_{k-1,k+1,k+1},
\]%
\[
\mathfrak{p}^{+}\text{Maxw}_{k+2,k,-(k+2)}\subset\text{Maxw}_{k+3,k+1,-(k+3)}%
\]%
\[
\mathfrak{p}^{-}\text{Maxw}_{k+2,k,-(k+2)}\subset\text{Maxw}_{k+1,k-1,-(k+1)}%
\]
and
\[
\mathfrak{p}^{+}\text{Maxw}_{k,k+2,-(k+2)}\subset\text{Maxw}_{k+1,k+2,-(k+3)}%
\]%
\[
\mathfrak{p}^{-}\text{Maxw}_{k,k+2,-(k+2)}\subset\text{Maxw}_{k-1,k+1,-(k+1)}%
.
\]
This proves the $G$-invariance of the indicated spaces. We next note that since
the multiplicities of the $K$-types of Maxw are all one, the Hermitian
form $\left\langle \ldots,\ldots\right\rangle $ restricted to each space
 {\rm Maxw}$_{k,l,m}$ is $0$, positive definite or negative definite. 
Thus Lemma 16 implies that the form is negative definite 
on ({\rm Maxw}$_{0,2}^{+})_{K}$ and ({\rm Maxw}$_{0,2}^{-})_{K}$.
 Recall that $\eta^{\ast}({\rm Maxw}_{0,2}^{+})_{K}=
({\rm Maxw}_{2,0}^{-})_{K}$ and $\eta^{\ast}
{\rm Maxw}_{0,2}^{-})_{K}=
({\rm Maxw}_{2,0}^{+})_{K}$. Furthermore $\eta$ is orientation reversing on $SU(2)$,
thus the form on $({\rm Maxw}_{2,0}^{+})_{K}$ and ({\rm Maxw}$_{2,0}^{-})_{K}$ is
positive definite. If one of the representations were not irreducible, then it
would have a finite-dimensional invariant subspace (by the above formulae).
Since these representations are all unitarizable and do not contain one-dimensional invariant subspaces the representations are all irreducible.
\end{proof}

\begin{remark}\emph{
The proof above implies that the last six inclusions above are all equalities.}
\end{remark}

We set $\pi_{i,j}^{\varepsilon}$ equal to the action of $G$ on Maxw$_{i,j}%
^{\varepsilon}$ for $\varepsilon=+,-$ and $i,j=0,2$ or $2,0$.

\begin{theorem}
The representations $\pi_{i,j}^{\varepsilon}$ for $\varepsilon=+,-$ and
$i,j=0,2$ or $2,0$ all have trivial infinitesimal character (that is equal to
the restriction to the center of $U(\mathfrak{g})$ of the augmentation
homomorphism to its center). Furthermore%
\[
H^{2}(\mathfrak{g,}K,{\rm Maxw}_{i,j}^{\varepsilon})=
\mathbb{C},\varepsilon=+,-\text{ and }i,j=0,2\text{ or }2,0.
\]

\end{theorem}

\begin{proof}
By the above, all four of the representations are the spaces of $C^{\infty}$
vectors of an irreducible unitary representation. Since the center of
$\mathfrak{g}$ acts trivially and
\[
\dim {\Hom}_{K}(\wedge^{2}\mathfrak{p,}M)=1
\]
for each $M$ as in the statement, it is enough to prove that the Casimir
operator corresponding to the trace form on $\mathfrak{g}$ acts by $0$ on each
of the representations (c.f.~[BW]). We will show that the center of the
enveloping algebra acts correctly. We look at the case Maxw$_{2,0}^{+}$ and
leave the other cases to the reader. By Lemma 18 $($Maxw$_{2,0}^{+})_{K}$ is a
highest weight module relative to the Weyl chamber (in $\varepsilon$ notation)
$\varepsilon_{1}>\varepsilon_{2}>\varepsilon_{3}>\varepsilon_{4}$. We
calculate the highest weight $\lambda=a\varepsilon_{1}+b\varepsilon
_{2}+c\varepsilon_{3}+d\varepsilon_{4}$. The representation factors through
the adjoint group so $a+b+c+d=0$. By definition of $F^{2,0,2}$
$a-b=2,c-d=0,c+d=2$. Solving the four equations yields $-2\varepsilon
_{2}+\varepsilon_{3}+\varepsilon_{4}$. $\rho$ for the chamber that we are studying
is $\frac{3}{2}\varepsilon_{1}+\frac{1}{2}\varepsilon_{2}-\frac{1}%
{2}\varepsilon_{3}-\frac{3}{2}\varepsilon_{4}$. Thus $\lambda+\rho=\frac{3}%
{2}\varepsilon_{1}-\frac{3}{2}\varepsilon_{2}+\frac{1}{2}\varepsilon_{3}%
-\frac{1}{2}\varepsilon_{4}=\sigma\rho$ for $\sigma$ the cyclic permutation
$(243)$.
\end{proof}

\medskip

Let $\theta$ be the Cartan involution of $G_{1}$ corresponding to $K$. In
light of the Vogan--Zuckerman theorem [VZ] this implies that there are four
$\theta$ stable parabolic subalgebras $\mathfrak{q}_{i,j}^{\varepsilon}$
$\varepsilon=+,-$ and $i,j=0,2$ or $2,0$, such that $({\rm Maxw}_{i,j}%
^{\varepsilon})_{K}$ is isomorphic with $A_{{\mathfrak q^\varepsilon_{i,j}}}(0)$ (c.f.~[BW]).

\begin{theorem}
$\mathfrak{q}_{2,0}^{+}$ is the parabolic subalgebra
\[
\{[x_{ij}]\in M_{4}(\mathbb{C})|x_{21}=x_{31}=x_{41}=0\},
\]%
\[
\mathfrak{q}_{0,2}^{+}=\{[x_{ij}]\in M_{4}(\mathbb{C})|x_{41}=x_{42}=x_{43}=0\},
\]
and $\mathfrak{q}_{0,2}^{-}=(\mathfrak{q}_{0,2}^{+})^{T},\mathfrak{q}_{2,0}%
^{-}=(\mathfrak{q}_{2,0}^{+})^{T}$.
\end{theorem}

\begin{proof}
If $\mathfrak{q}$ is a $\theta$-stable parabolic subalgebra with nilpotent
radical $\mathfrak{u}$ and if $\mathfrak{u}_{n}=\mathfrak{u}\cap\mathfrak{p} $,
we set $2\rho_{n}(\mathfrak{q})(h)= \rm{tr}\rm{ad}(h)_{|}\mathfrak{u}_{n}$ for
$h\in\mathfrak{h}$ the Cartan subalgebra of diagonal matrices. One checks that
if the parabolics are given as in the theorem, then $2\rho_{n}(\mathfrak{q}_{2,0}%
^{+})$ is the highest weight of $F^{2,0,2}$, $2\rho_{n}(\mathfrak{q}%
_{0,2}^{+})$ is that of $F^{0,2,2},2\rho_{n}(\mathfrak{q}_{2,0}^{-})$ is the
lowest weight of $F^{2,0,-2}$, and $2\rho_{n}(\mathfrak{q}_{0,2}^{-})$ is the
lowest weight of $F^{0,2,-2}$. Now the result follows from the main theorem
(c.f. [W1]).
\end{proof}

\begin{remark}\emph{
This result implies that if we consider the open orbit $X$ of $U(2,2)$ in
$\mathbb{P}^{3}(\mathbb{C})$ that has no non-constant holomorphic functions, then there are two
holomorphic line bundles and two anti-holomorphic line bundles such that their
degree $1$ smooth sheaf cohomology yields the four versions of solutions of
Maxwell's equations. The relation between the two realizations is the Penrose--Twistor transform.}
\end{remark}

Another realization of these representations is in [GW]. The observation here
is that $SU(2,2)$ is the quaternionic real form of 
$SL(4,\mathbb{C})$. Since $K\cap SU(2,2)=S(U(2)\times U(2))$ there are two invariant
quaternionic structures on $SU(2,2)/K\cap SU(2,2)$ and on each of two line
bundles on which the methods of [GW] apply. Each yields an element of the
\textquotedblleft analytic continuation\textquotedblright\ of the
corresponding quaternionic discrete series which is our third realization. We
will allow the interested reader to check these assertions.

Finally, these representations appear in Howe duality between $U(2,2)$ and the
$U(1)$ in its center. The details are worked out in [BW], VIII, 2.10--2.12 in
the notation therein, and  the pertinent representations are $V_{2}$ and $V_{-2}$
yielding \rm{Maxw}$_{2,0}^{+})_{K}$ and \rm{Maxw}$_{0,2}^{+})_{K}$. The other two
constituents are gotten by duality.

\section{A dual pair in $PSU(2,2)$}

We continue to consider $U(2,2)$ as the group $G_{1}$ in Section 2 and
$PSU(2,2)$ the quotient by the center. Let $S$ be the image of the subgroup of
all matrices of the form%
\[
\left[
\begin{array}
[c]{cc}%
aI & bI\\
\overline{b}I & \overline{a}I
\end{array}
\right]
\]
with $|a|^{2}-|b|^{2}=1$. Then $S$ is isomorphic with $PSU(1,1)$. The image of
the subgroup of $PSU(2,2)$ that centralizes every element of $S$ is the image
of the group of all block $2\times2$ diagonal elements of $G_{1}$, $C$ which
is isomorphic with $SO(3)$. These two groups form a reductive dual pair (the
commutant of $C$ is $S$). Before we analyze this pair, we will indicate its
relationship with time in Minkowski space. If we write coordinates in Section
2 as $x_{1},x_{2}.x_{3}$ and $x_{4}=t$, then we have according to our rules
\[
(0,0,0,t)\longmapsto\frac{(1+it)}{1-it}I.
\]
In terms of the linear fractional action of $U(2,2)$ on $U(2)$ this
corresponds to%
\[
\left[
\begin{array}
[c]{cc}%
\frac{1+it}{\sqrt{1+t^{2}}}I & 0\\
0 & \frac{1-it}{\sqrt{1+t^{2}}}I
\end{array}
\right]  .I.
\]
That is, the time axis is the orbit of $K\cap S$ with $-I$ deleted.

As we observed, the center of $U(2,2)$ acts trivially on Maxw and thus we can
consider the action of the group $C\times S$ through $CS$ on Maxw. Let
$\mathcal{H}^{k}$ denote the representation of $SO(3)$ on the spherical
harmonics of degree $k$. We denote by $D_{2k}$,$k\in
\mathbb{Z}
-\{0\},$ the discrete series of $S$ ($D_{2k}$ has $K$-types $2k+2$sgn$(k)
\mathbb{Z}_{\geq0}$).

\begin{theorem}
As a representation of $C\times S$, (\rm{Maxw}$_{2,0}^{+})_{K}$ and
(\rm{Maxw}$_{0,2}^{+})_{K}$ are equivalent with%
\[
\bigoplus_{k\geq1}\mathcal{H}^{k}\otimes D_{-(2k+2)}%
\]
and $(\rm{Maxw}_{2,0}^{-})_{K}$ and $(\rm{Maxw}_{0,2}^{-})_{K}$ are equivalent with%
\[
\bigoplus_{k\geq1}\mathcal{H}^{k}\otimes D_{2k+2}.
\]

\end{theorem}

\begin{proof}
By considering the $K$-types we see that each of the modules  has finite $K \cap CS$-multiplicities. We will give details for $M=($Maxw$_{2,0}^{+})_{K}%
$. Noting that $S\cap K$ is
\[
\left\{  \left[
\begin{array}
[c]{cc}%
aI & \\
& a^{-1}I
\end{array}
\right]  |\left\vert a\right\vert =1\right\}
\]
we see that the characters that appear for $S\cap K$ on $M$ are $a^{-4}%
,a^{-6},...$ each with finite multiplicity. Thus, as a representation of $S$,
$M$ is a direct sum of highest weight modules. It is therefore enough to check
that the character of $M$ as a $CS\cap K$ module is correct. We note that in
the formulation of the $K$-types, we can look upon $CS\cap K$ as the image of
the group of all matrices%
\[
\left[
\begin{array}
[c]{cc}%
u(y) & 0\\
0 & u(y)x^{-2}%
\end{array}
\right]
\]
with $u(y)=\left[
\begin{array}
[c]{cc}%
y & 0\\
0 & y^{-1}%
\end{array}
\right]  $ and $|x|$ and $|y|=1$. Set $t(x,y)$ equal to this element. Then if
$Ch(V)$ denotes the character of a $CS\cap K$ module and $\chi_{k}$ denotes
the character of $SU(2)$ on $S^{k}(\mathbb{C}^{2})$,  
we have as a formal sum
\[
Ch(M)=\sum_{k=0}^{\infty}\chi_{k+2}(u(y))\chi_{k}(u(y))x^{-2k-4}%
\]
replacing $x$ with $x^{-1}$ and plugging in $\chi_{k}(u(y))=y^{k}%
(1+y^{-2}+\ldots+y^{-2k})$ and summing we have%
\[
\frac{x^{4}(y^{4}-x^{2}y^{2}+y^{2}+1)}{(1-x^{2})(y^{2}-x^{2})(1-x^{2}y^{2})}.
\]
If we multiply by $1-x^{2}$ and expand in powers of $x$, we see that the series
is%
\[
x^{4}\sum_{k=0}^{\infty}\chi_{2k+2}(u(y))x^{2k}.
\]
This implies the result in this case since the character of $D_{-2k}$ is
\[
\frac{x^{-2k}}{1-x^{-2}}\left(  \text{resp. }\frac{x^{2k}}{1+x^{2}}\right)
\]
if $k>0$ (resp.$k<0$). The other highest weight case is exactly the same. The
cases (\rm{Maxw}$_{2,0}^{-})_{K}$ and (\rm{Maxw}$_{0,2}^{-})_{K}$ are handled in the
same way but with the powers of $x$ inverted.
\end{proof}

\section{Plane wave solutions and (degenerate) \\Whittaker vectors}

Recall that a plane wave solution to Maxwell's are those in the following form
\[
\mathbf{E}=e^{i(z,x)}\mathbf{E}_{o}%
\]
and%
\[
\mathbf{H}=e^{i(z,x)}\mathbf{H}_{o}%
\]
with $\mathbf{E}_{o}$ and $\mathbf{H}_{o}$ constant vectors, $z=(u_{1}%
,u_{2},u_{3},\omega)\in\mathbb{R}^{4}$,
 $\mathbf{u}=(u_{1},u_{2},u_{3})$, $x=(x_{1},x_{2},x_{3},t)\in\mathbb{R}^{4}$,
 and as before, $(z,x)=-\sum u_{i}x_{i}+\omega t$. To satisfy Maxwell's
equations we must have $\mathbf{u\cdot E}_{o}=0,\mathbf{u\cdot H}_{o}=0$ and
\[
\nabla\times\mathbf{E}=\frac{\partial}{\partial t}\mathbf{H,}\nabla
\times\mathbf{H}=-\frac{\partial}{\partial t}\mathbf{E}.
\]
The first implies that $\mathbf{u}\times\mathbf{E}_{o}=-\omega\mathbf{H}_{o}$
and the second $\mathbf{u}\times\mathbf{H}_{o}=\omega\mathbf{E}_{o}$.  Thus, if
the solution is non-constant and $\omega>0$ (resp. $\omega<0$) $\frac
{\mathbf{u}}{\left\Vert \mathbf{u}\right\Vert },\frac{\mathbf{H}_{o}%
}{\left\Vert \mathbf{H}_{o}\right\Vert },\frac{\mathbf{E}_{o}}{\left\Vert
E_{o}\right\Vert }$ (resp. $\frac{\mathbf{u}}{\left\Vert \mathbf{u}\right\Vert
},\frac{\mathbf{H}_{o}}{\left\Vert \mathbf{H}_{o}\right\Vert },\frac
{-\mathbf{E}_{o}}{\left\Vert E_{o}\right\Vert }$) is an orthonormal basis of 
$\mathbb{R}^{3}$ 
obeying the \textquotedblleft right-hand screw law\textquotedblright.
Putting these equations together yields $\omega^{2}=\left\Vert \mathbf{u}%
\right\Vert ^{2}$, that is, $z$ is isotropic in $\mathbb{R}^{1,3}$.
That is on the null light cone. This fits with the fact that if
$\mathbf{E},\mathbf{H}$ form a solution to Maxwell's equations, then their
individual components satisfy the wave equation.

The purpose of this section is to interpret the plane wave solutions in the
context of the four representations we have been studying. In the last few
sections we have been emphasizing the realization of $U(2,2)$ which we have
denoted by $G_{1}$. We now revert to the form $G$ since the embedding of
Minkowski space into $U(2)$ is clearer for that realization. Recall that the
embedding is implemented in two stages;  we map $(x_{1},x_{2},x_{3},x_{4})$ to
\[
\left[
\begin{array}
[c]{cc}%
I & 0\\
Y & I
\end{array}
\right]  ,Y=\left[
\begin{array}
[c]{cc}%
x_{4}+x_{3} & x_{1}+ix_{2}\\
x_{1}-ix_{2} & x_{4}-x_{3}%
\end{array}
\right]  .
\]
We denote the image group by $\overline{N}$ (note $Y^{\ast}=Y$). Then we
consider the image of $I$ under the  action of $G$ on $U(2)$ by linear
fractional transformations. In this context to simplify the form subgroup $S$ we must
apply a Cayley transform. If we set
\[L=\frac{1}{\sqrt{2}}\left[
\begin{array}
[c]{cc}%
I & iI\\
I & -iI
\end{array}
\right],
\]
then the transform is given by
\[
\sigma^{-1}(g)=L^{-1}gL.
\]
Thus, in this incarnation, we take $S$ to be the set of elements of $PSU(2,2)
$ as $G/\rm{Center(G)}$ in the form
\[
\left[
\begin{array}
[c]{cc}%
aI & bI\\
cI & dI
\end{array}
\right]
\]
with $\operatorname{Im}a\overline{b}=0,a\overline{b}-b\overline{c}%
=1,\operatorname{Im}c\overline{d}=0$. This implies that the image of
$S\cap\overline{N}$ is the set of matrices
\[
\left[
\begin{array}
[c]{cc}%
I & 0\\
cI & I
\end{array}
\right]
\]
with $c \in \mathbb{R}$. 
Let $z=(k_{1},k_{2},k_{3},\omega)$ be, as above, an element in the light
cone. Then considering $z$, as giving a linear functional on $\Lie(\overline
{N})$,  this functional restricted to$\left[
\begin{array}
[c]{cc}%
0 & 0\\
xI & 0
\end{array}
\right]  \in \Lie(S\cap\overline{N})$ is given by the coefficient of $q$ in the
quadratic polynomial%
\[
\frac{1}{2}\det\left[
\begin{array}
[c]{cc}%
\omega+k_{3}+qx & k_{1}+ik_{2}\\
k_{1}-ik_{2} & \omega-k_{3}+qx
\end{array}
\right]
\]
which is $\omega x$. We will record this as a lemma since we will need to
apply it later in this section.

\begin{lemma}
If $z=(k_{1},k_{2},k_{3},\omega)$ is in the null light cone and if we consider $z$
as a linear functional on $\Lie(\overline{N})$ (as above), then its value on%
\[
\left[
\begin{array}
[c]{cc}%
0 & 0\\
xI & 0
\end{array}
\right]
\]
is $\omega x$.
\end{lemma}

We now relate these plane wave solutions to the representations Maxw$_{\alpha
,b}^{\pm}$ with $(\alpha,\beta)\in\{(2,0),\{0,2)\}$. Let $H_{\alpha,\beta
}^{\pm}$ be the corresponding Hilbert space completions of these smooth
Fr\'{e}chet representations of moderate growth. Then the Maxw$_{\alpha,b}%
^{\pm}$ are the spaces of $C^{\infty}$ vectors. Fix one of these
representations and denote it by $(\pi,H)$. Let $H_{K}$ be the space of
$K$-finite vectors and let $\left\langle v,w\right\rangle _{k}=\left\langle
(I+C_{K})^{kd}v,w\right\rangle $ for $v,w\in H_{K}$, $\left\langle
\ldots,\ldots\right\rangle $ the unitary structure and $d=\dim K$. Let $H^{k}$ be
the Hilbert space completion of $H_{K}$ with respect to $\left\langle
\ldots,\ldots\right\rangle _{k}$. Then $H^{0}=H$ and the maps $T_{k,l}%
:H^{l}\rightarrow H^{k}$ for $k>l\geq0$ that are the identity on $H_{K}$ are
nuclear (see [GV] for the definition). Thus we have $\cap_{k\geq0}%
H^{k}=H^{\infty}$. We set $H^{-k}$ equal to the dual Hilbert space to $H^{k}$.
Thus $(H^{\infty})^{\prime}=\cup_{k\geq0}H^{-k}$. We thus have a rigged
Hilbert space in the sense of [GV], Chapter 1. We now consider $(\pi,H)$ as a
unitary representation of $\overline{N}$. We have (c.f.~[W1], Theorem 14.10.3)
a direct integral decomposition of $(\pi,H)$ as a representation of
$\overline{N}$ given as follows:%
\[
\int_{\text{supp}(\pi)\subset\widehat{\overline{N}}}
\mathbb{C}_{\chi}\otimes W_{\chi}d\tau(\chi)
\]
with supp$(\pi)$ defined as in [W1, Volume 2, p.~337], $\tau$ a Borel measure on
supp$(\pi)$,  and $\chi\rightarrow W_{\chi}$ a Borel measurable Hilbert vector
bundle over supp$(\pi)$. In this case supp$(p)$ is just the support of $\tau$
as a distribution. The theorem of Gelfand--Kostyuchenko (c.f.~[GV], p.~117
Theorem 1$^{\prime}$ and the remarks at the end of Subsection 1.4.4) implies 
that there exists, $k>0$ and a family of nuclear operators $T_{\chi}%
:H^{k}\rightarrow W_{\chi}$ such that if $u\in H^{k}$, then $u(\chi)=T_{\chi
}(u)$ for $\tau$-almost all $\chi$. This implies that for all $\chi$ and all
$u\in H^{k}$,
\[
T_{\chi}(\pi(\overline{n})u)=\chi(\overline{n})T_{\chi}(u).
\]
If $T_{\chi}$ is always $0$ then $H=0$. We see that there is a subset $A$ of
supp$(\pi)$ of full measure with $T_{\chi}\neq0$ for $\chi\in A$. By definition of the $H^{k}$
we see that if $\chi\in A$, then there exists $\lambda\in W_{\chi}^{\prime}$
such that $\lambda\circ T_{\chi}\neq0$. We have for this choice%
\[
\lambda\circ T_{\chi}(\pi(\overline{n})u)=\chi(\overline{n})\lambda\circ
T_{\chi}(u)
\]
for all $\overline{n}\in\overline{N}$ and all $u\in H^{\infty}$. If $\chi
\in\widehat{\overline{N}}$, then set $$Wh_{\chi}(\pi)=\{\mu\in(H^{\infty
})^{\prime}|\mu\circ\pi(\overline{n})=\chi(\overline{n})\mu\}.$$ We have with
this notation

\begin{theorem}
Let $H^{\infty}$ be one of $\rm{Maxw}_{2,0}^{\pm}$ or $\rm{Maxw}_{2,0}^{\pm}$. Then

1. $\dim Wh_{\chi}(\pi)\leq3.$

2. Writing
\[
\chi_{Z}\left(  \left[
\begin{array}
[c]{cc}%
I & 0\\
Y & I
\end{array}
\right]  \right)  =e^{i(Z,Y)}%
\]
with $Z^{\ast}=Z$ (this describes all possible $\chi$). Then $Wh_{\chi_{Z}%
}(\pi)=0$ if $\det(Z)\neq0$ (i.e., the support of $\pi$ is contained in the
null light cone). Furthermore, the support of $\pi$ is the closure of a single
orbit of the action of $GL(2,\mathbb{C})$ on 
$\{Z\in M_{2}(\mathbb{C})|Z^{\ast}=Z\}$ given by $g\cdot Z=gZg^{\ast}$.

3. Both $\pi_{2,0}^{+},\pi_{0,2}^{+}$ have support equal to the set of
all $\chi_{Z}$ with $\det Z=0$ and $\rm{tr}Z\leq0$ (negative light cone) and both
$\pi_{2,0}^{-},\pi_{2,0}^{-}$ have support equal to the set of all $\chi_{Z}$
with $\det Z=0$ and {\rm tr}$Z\geq0$ (positive light cone).
\end{theorem}

\textbf{Note: }The proof of this result will use earlier work of the authors
[K],[W2],[W3].

\bigskip

\begin{proof}
We will do all the details for Maxw$_{2,0}^{+}$. The case Maxw$_{0,2}^{+}$ is
essentially the same. The cases Maxw$_{2,0}^{-}$ and Maxw$_{0,2}^{-}$ are done
with an interchange of $\mathfrak{p}^{\pm}$ with $\mathfrak{p}^{\mp}$. Set
$\mathfrak{q=k}\oplus\mathfrak{p}^{+}$ where $\mathfrak{k}=\Lie(K)_{\mathbb{C}}$. 
We consider $F^{2,0,2}$ to be the $\mathfrak{q}$-module with
$\mathfrak{k}$ acting through the differential of the $K$ action and
$\mathfrak{p}^{+}$ acting by $0$. Set
\[
N(F^{2,0,2})=U(\mathfrak{g})\otimes_{U(\mathfrak{q)}}F^{2,0,2}.
\]
Then Lemma 18 implies that we have a surjective $(\mathfrak{g},K)$-module
homomorphism
\[
p:N(F^{2,0,2})\rightarrow{\rm{Maxw}}_{2,0}^{+})_{K}.
\]
Thus $p^{\ast}Wh_{\chi}(\pi_{2,0}^{+})$ is contained in
\[
\mathcal{W}_{\chi}=\{\lambda\in {\Hom}_{\mathbb{C}}(N(F^{2,0,2}),
\mathbb{C})|\lambda(Yv)=d\chi(Y)\lambda(v),Y\in \Lie(\overline{N})\}.
\]
Since $p^{\ast}:Wh_{\chi}(\pi_{2,0}^{+})\rightarrow\mathcal{W}_{\chi}$ is
injective ($({\rm{Maxw}}_{2,0}^{+})_{K}$ is dense in ${\rm{Maxw}}_{2,0}^{+}$) the dimension
estimate will follow if we prove that $\dim\mathcal{W}_{\chi}=3$. For this we
use the observation in [W3] that if
\[
L=\left\{  \left[
\begin{array}
[c]{cc}%
g & 0\\
0 & (g^{-1})^{\ast}%
\end{array}
\right]  |g\in GL(2,\mathbb{C})\right\},
\]
then $\Lie(L\overline{N})_{
\mathbb{C}}$ 
and $\mathfrak{q}$ are opposite parabolic subalgebras (i.e., $\Lie(L\overline
{N})_{\mathbb{C}}\cap\mathfrak{q}$ is a Levi factor of both parabolic subalgebras). This
implies that $N(F^{2,0,2})$ is a free $\Lie(\overline{N})_{\mathbb{C}}$ 
module on $\dim F^{2,0,2}=3$ generators. This clearly implies that
$\dim\mathcal{W}_{\chi}=3$.

We now note  if $m=\left[
\begin{array}
[c]{cc}%
g & 0\\
0 & (g^{-1})^{\ast}%
\end{array}
\right]  $ then $m\left[
\begin{array}
[c]{cc}%
I & 0\\
Y & I
\end{array}
\right]  m^{-1}=\left[
\begin{array}
[c]{cc}%
I & 0\\
g\cdot Y & I
\end{array}
\right]  $. 
Set $\overline{n}=\overline{n}(Y)=\left[
\begin{array}
[c]{cc}%
I & 0\\
Y & I
\end{array}
\right]  $. If $\lambda\in Wh_{\chi_{Z}}(\pi_{2,0}^{+})$, then
\[
\lambda\circ\pi_{2,0}^{+}(m)(\pi_{2,0}^{+}(\overline{n})u)=\lambda(\pi
_{2,0}^{+}(m)\pi_{2,0}^{+}(\overline{n})u)
\]%
\[
=\lambda(\pi_{2,0}^{+}(m\overline{n}m^{-1})\pi_{2,0}^{+}(m)u)=\chi
_{Z}(m\overline{n}m^{-1})\lambda\circ\pi_{2,0}^{+}(m)(u)
\]%
\[
=\chi_{g^{-1}\cdot Z}(\overline{n})\lambda\circ\pi_{2,0}^{+}(m)(u).
\]
Thus $\lambda\circ\pi_{2,0}^{+}(m)\in Wh_{\chi_{g^{-1}Z}}(\pi_{2,0}^{+}) $.
This implies that supp($\pi_{2,0}^{+}$) is a union of orbits.

We now observe  using results in [W3] that we can show that if $Wh_{\chi_{Z}%
}(\pi_{2,0}^{+})\neq0$ with $\det Z\neq0$, then the $GK$-dimension of
$({\rm Maxw}_{2,0}^{+})_{K}$ is at least $\dim\overline{N}$=4. Indeed, this
condition implies that the dual module to ${{\rm{Maxw}}}_{2,0}^{+})_{K}$ contains an
irreducible submodule of $GK$-dimension 4. However as a $\mathfrak{p}^{-}%
$-module $({\rm{Maxw}}_{2,0}^{+})_{K}$ is graded with the $k$-th level of the grad
being isomorphic with the $K$-module $F^{k+2,k,k+2}$ whose dimension is
$(k+3)(k+1)$. Thus the $GK$-dimension of $({\rm{Maxw}}_{2,0}^{+})_{K}$ is $3.$ Thus we
must have $\det Z=0$. To complete the proof, it is enough to prove 3. Since
there are three orbits of $L$ in set of all $Z$ with $\det Z=0$:%
\[
\mathcal{O}^{+}=\{Z|\det Z=0,\text{tr}Z>0\},\mathcal{O}^{+}=\{Z|\det
Z=0,\text{tr}Z<0\},\{0\}.
\]
Thus we must prove that if $Wh_{\chi_{Z}}(\pi_{2,0}^{+})\neq0$ and $Z\neq0$
then tr$Z<0$. This follows from Theorem 17 (p.~306) in [W2] (in this reference
the roles of $\chi$ and $\chi^{-1}$ are reversed). We can also prove the
result directly using Lemma 25. Let $\lambda\neq0$ be an element of
$Wh_{\chi_{Z}}(\pi_{2,0}^{+})$ with $Z\neq0$. Then set $\eta=d\chi_{Z}$. If
\[
Z=\left[
\begin{array}
[c]{cc}%
\omega+k_{3} & k_{1}+ik_{2}\\
k_{1}-ik_{2} & \omega-k_{3}x
\end{array}
\right],
\]
then tr$Z=2\omega$. In Lemma 25 we saw that
\[
\eta\left[
\begin{array}
[c]{cc}%
0 & 0\\
xI & 0
\end{array}
\right]  =\omega x.
\]
This if $S$ is as in the previous section and $\nu=\chi_{|N\cap S}$ then in
the notation of [W2] Theorem 2, $r_{\nu}=\omega$. Thus that theorem (also see
[K]) implies that (taking into account the reversal of signs mentioned above)
$r_{\nu}<0$. This completes the proof.
\end{proof}

\begin{remark}\emph{
The above results imply that in the steady state solutions to Maxwell's
equations the plane wave solutions should be looked upon as 2-currents. That
is, using $\gamma$, they are distributions on $\rm{Maxw}$.}
\end{remark}

\section{Conclusion}

In this paper we have shown (Theorem 19) that the solutions to Maxwell's
equations that extend to the conformal compactification of Minkowski space
break up into four irreducible unitary representations: $2$ positive energy and
$2$ negative energy (in the sense of lowest weight or highest weight
respectively) as a representation of the conformal group of the wave equation.
The support of each of the representations is either the forward or backward 
null light cone. The ones with positive energy (Theorem 26) yield only positive
frequencies and thus according to Planck's law have positive energy in the
sense of field theory. We observe that solutions that extend to the total
compactification (i.e., steady state) can have their frequency spectrum
constrained to a narrow band yielding background radiation that indicates a
temperature of $2.7$ degrees Kelvin (or any other temperature for that
matter). This says that although the actual steady state models to the
universe that fit the astronomical observations are complicated this work
indicates that there can be a background radiation that fits the measurements
that is not the outgrowth of an initial very high temperature source. We would
also like to point out that the big bang models for the universe have had
difficulty fitting the observations also, leading to theories involving
inflation and the return of the cosmological constant. An interesting
alternative which is not unrelated to this paper can be found in [P] where the
conformal structure is emphasized and time does almost cycle. There is also
the chronometric theory of [S].

We have also shown that these representations fit in larger contexts. However,
although many of the results in this paper generalize to $SO(n,2)$, the
beautiful geometric structures that appear in the case $n=4$ do not.

\bigskip

\begin{center}
{\Large References}
\end{center}

\bigskip

\noindent[ BW] A. Borel and N. Wallach, \emph{Continuous cohomology, discrete
subgroups, and representations of reductive groups,} Second edition.
Mathematical Surveys and Monographs, 67. American Mathematical Society,
Providence, RI, 2000.

\smallskip

\noindent[ EW] T. Enright and N. Wallach, Embeddings of unitary highest
weight representations and generalized Dirac operators, \emph{Math. Ann.}~{\bf307} (1997), 627--646.

\smallskip

\noindent[ GV] I. M. Gelfand and N. Y.Vilenkin,\emph{ Generalized Functions Vol. 4: Applications of Harmonic Analysis}, Academic Press, 1964.

\smallskip

\noindent[HBN] F. Hoyle, G. Burbidge, J.V. Narlikar, A quasi-steady
state cosmological model with creation of matter, \emph{The Astrophysical Journal}
{\bf410}(1993), 437--457.

\smallskip

\noindent[ GW] B. Gross and N. Wallach, On quaternionic discrete series
representations, and their continuations,\emph{ J. Reine Angew. Math.}, {\bf 481}(1996), 73--123.

\smallskip

\noindent[ HSS] M. Hunziker, M. Sepanski and R. Stanke, Conformal
symmetries of the wave operator, to appear.

\smallskip

\noindent[ K] B. Kostant, On Laguerre polynomials, Bessel functions,
Hankel transform and a series in the unitary dual of the simply-connected
covering group of $SL(2,\mathbb{R})$,  \emph{Represent. Theory} {\bf 4} (2000), 181--224.

\smallskip

\noindent[ P] Roger Penrose, \emph{Cycles of Time}, Alfred Knopf, New York, 2011.

\smallskip

\noindent[ S] Irving Segal, Radiation in the Einstein universe and the
cosmic background, \emph{Physical Review D} {\bf 28}(1983), 2393--2402.

\smallskip

\noindent[ VZ] D. Vogan and G. Zuckerman, Unitary representations and
continuous cohomology, \emph{Comp/ Math.} {\bf53} (1984), 51--90.

\smallskip

\noindent[ W1] Nolan R. Wallach,\emph{ Real reductive groups I,II}, Academic
Press, New York, 1988,1992.

\smallskip

\noindent[ W2] N. Wallach, Generalized Whittaker vectors for holomorphic
and quaternionic representations, \emph{Comment. Math. Helv.} {\bf78} (2003), 266--307.

\smallskip

\noindent[ W3] N.R. Wallach, Lie algebra cohomology and holomorphic
continuation of generalized Jacquet integrals, \emph{Advance Studies in Pure Math} {\bf14} (1988),123--151.

\end{document}